\documentclass[11pt]{amsart}

\usepackage[table]{xcolor}
\usepackage{graphicx} 
\usepackage{natbib}
\usepackage{tikz}
\usetikzlibrary{graphs}
\usepackage{url}
\usepackage{hyperref}
\usepackage{amsmath,amsthm,amssymb}
\usepackage{a4wide}
\usepackage{dsfont}
\usepackage{todonotes}
\usepackage{url}
\usepackage{enumitem}
\usepackage{placeins}
\usepackage{float}
\usepackage{geometry}
\newtheorem{Theorem}{Theorem}
\newtheorem{lem}{Lemma}
\newtheorem{cor}{Corollary}
\newtheorem{bem}{Remark}
\newtheorem{bsp}{Example}

\newcommand{\N}{\mathbb{N}}
\newcommand{\I}{\mathbb{I}}
\newcommand{\R}{\mathbb{R}}
\newcommand{\E}{\mathbb{E}}

\newcommand{\Prob}{\mathbb{P}}

\newcommand{\Var}{\textnormal{Var}}

\newcommand{\ignore}[1]{}
\def\multiset#1#2{\ensuremath{\left(\kern-.3em\left(\genfrac{}{}{0pt}{}{#1}{#2}\right)\kern-.3em\right)}}

\title{A two-sample test based on averaged Wilcoxon rank sums over interpoint distances}
\author{Annika BETKEN$^*$}
\author{Aljosa MARJANOVIC$^*$}
\author{Katharina PROKSCH$^*$}
\thanks{$^*$University of Twente, Drienerlolaan 5, 7522 NB Enschede, Netherlands}

\begin{document}

\pagestyle{plain}
	\begin{abstract}
		\noindent
		An important class of two-sample multivariate homogeneity tests is based on identifying   differences between the distributions of interpoint distances. While generating distances from point clouds offers a straightforward and intuitive way for dimensionality reduction, it also introduces dependencies  to the resulting distance samples. We propose a simple test based on Wilcoxon's rank sum statistic for which we prove asymptotic normality under the null hypothesis and fixed alternatives under mild conditions on the underlying distributions of the point clouds. Furthermore, we show  consistency of the test and derive a variance approximation that allows to construct a computationally feasible, distribution-free test with good finite sample performance. The power and robustness of the test 
		for high-dimensional data and low sample sizes
		is demonstrated by numerical simulations. Finally, we apply the proposed test to case-control testing on microarray data in genetic studies, which is considered   a  notorious case  for a high number of variables and low sample sizes. 
	\end{abstract}
	\maketitle	

	\section{Introduction}
	Two-sample homogeneity testing, i.e.\ inferring whether two sets of observations stem from the same underlying distribution or not, is a standard statistical problem. Commonly used tests for  univariate observations are the parametric t-test, which assumes normally distributed data, and non-parametric alternatives based on the comparison of empirical distribution functions such as the Kolmogorov-Smirnov and Cram\'er-von Mises tests. Another robust, non-parametric test is given by the Wilcoxon rank sum test, which is based on the sum of the ranks of one set of observations within the pooled set of observations. 
	Major strengths of Wilcoxon's test are that it does not pose strong assumptions on the underlying distributions of samples. It may therefore be valid when the assumptions of other tests are not met, while, at the same time,  being asymptotically distribution free. Since the test exclusively uses information on the ranks of the data, it provides robustness against outliers, making it suited to deal with data from  heavy-tailed distributions. In the context of time series analysis, two-sample homogeneity tests are often interpreted as solutions to  change-point problems. Against this background, the Wilcoxon test has been well-studied for independent (\cite{gombay1998rank}), short-range dependent(\cite{gerstenberger2021robust}), as well as long-range dependent (\cite{dehling2013non}, \cite{dehling2017},  \cite{betken2016testing}) time series. 
	\\
	While parametric tests are readily extendable to multivariate models (for example, MANOVA procedures or Hotelling's $T^2$ test), these extensions are sensitive to high-dimensionality of observations and  may not keep their nominal level or may even dramatically lose power if the assumptions posed on the data-generating distributions are not met. Since there is no natural order on $\R^D$ for $D>1$, there is no unique, straightforward way to 
	define ranks for multivariate observations and, as a consequence thereof, to
	extend Wilcoxon's test in a way that it can be used in a multivariate setting. 
	A multitude of different approaches to define multivariate ranks on the basis of statistical  depth  exist in the literature.
	These approaches are  as numerous as the  underlying concepts of data depth such as, e.g., Mahalanobis depth  (see  \cite{Mahalanobis}), halfspace depth \cite{tukey1975mathematics}, Oja's depth \cite{oja1983descriptive}, or simplicial depth \cite{LiuSimplicial}, and  have  been used in the context of statistical hypothesis testing
	(see, e.g., \cite{LiuSingh93}, \cite{Small12}, \cite{Shi23}, \cite{Oja94}, or \cite{Oja99}).
	Recent approaches to multivariate ranks and data depth as well as their application in two sample tests use the concept of optimal transport by defining the ranks of a sample relative to a reference distribution from the uniform family \cite{chernozhukov2017monge,hallin2021distribution,deb2023multivariate}.
	All these approaches pose assumptions on the shapes of the underlying distributions or their supports and may become  computationally too intense in high dimensions or for too large data clouds. Moreover, an immediate or unique generalization to more complex models, e.g. for functional data analysis 
	does not exist. 
	In such contexts, a variety of depth concepts has been proposed in the literature (see \cite{gijbels2017general} for a discussion of different functional depths and their properties, or \cite{geenens2023statistical} for a general metric approach to depth).
	For a detailed discussion see also  \cite{Gnettner23}.
	
	On the other hand, transforming data in a multivariate point cloud or in complex spaces to its one-dimensional depth information suggests to take the general route to use some form of dimensionality reduction as a pre-processing step before applying a well-established one-dimensional method to the transformed data. A widely applicable approach is the use of  interpoint distances --or dissimilarities-- as meaningful notions of (dis)similarity exist in very general contexts.\\
	In two sample homogeneity testing, interpoint distances  within and between samples can be used as the distinguishing quantity to base a test on. The conditions for the equivalence of testing for equal distributions in the multivariate case and univariate testing for the equality of distance (or, more generally, point disimilarity functions) distributions are provided by the theorem of Maa et al. \cite{Maa96} and its adjustment by Montero-Manso and Villar  \cite{montero19}. There have been many approaches to comparing distance distributions in the literature. The class of interpoint distance based tests include energy statistics test by \cite{szekely04,baringhaus2004new}, the maximum mean discrepancy (MMD) based test by Gretton et al. \cite{gretton06}, k-nearest neighbors based tests of Henze \cite{henze88} and rank based generalization of Wilcoxon's test proposed by \cite{Liu22}. The latter, and similarly the energy distance and MMD tests, correspond to a degenerate $U$-statistics whose null distribution is therefore a quadratic form of Gaussian random variables, which necessitates the estimation of p-values by bootstrapping procedures. \\
	Distance-based approaches allow for analysis of data in general metric spaces, even if the data are modelled as metric measure spaces themselves, which is a commonly used model for general object comparisons as in \cite{Brecheteau,gellert2019substrate,Weitkamp}. 
	A general overview and further related interpoint distance based tests are given in \cite{montero19}. \\
	
	\noindent
	In this paper, we derive a two-sample homogeneity test based on interpoint distances.
	Our proposed test statistic is effectively a $U$-statistic-like estimate of $\Prob (d(X_1,X_2) < d(X_3,Y))$, where $d$ is a distance and $X_1,X_2,X_3\sim F_X, Y\sim F_Y$ are independent random variables on a common metric space $(\mathcal{X},d)$.
	The  distance samples used by our test are a substantial subset of all pair-wise distances with a particular dependence structure. We show the connection of the test statistic to the averaged Wilcoxon statistics over independent subsampling of the interpoint distances, which allows to derive our asymptotic theory. The averaged  Wilcoxon rank sum statistics test was introduced by Datta and Satten \cite{Datta05} for paired data under a simpler dependence structure, namely subsampling over pairwise independent dependency clusters, based on the resampling idea of Hoffman et al. \cite{Hoffman01}. We show the asymptotic normality of the test statistic under the null hypothesis and fixed alternatives under mild conditions on primary point-cloud distributions by applying a central limit theorem for locally dependent random variables, described by a dependency graph \cite{Austern22}. While the asymptotic distribution of the test statistic  could as well be approached  by asymptotic theory for $U$-statistics, the direct application of this theorem, rather than its corollary regarding the convergence of the related $U$-statistic or equivalent statements, allows us to consider specific dependence structure of occuring interpoint distance comparisons, which remains invariant between the null hypothesis and alternatives. In contrast to the Wilcoxon test based on independent data the variance in the limit distribution of our test statistic depends on the distribution of the point clouds and is hence not distribution-free. However, the variation in this variance is only very mild (see Section \ref{sec:numerical P1}). 
	In order to construct a test which is readily applicable, we 
	further derive an asymptotic estimate of the variance of the test statistic under the null hypothesis  and provide a distribution free  upper bound on this estimate that holds without posing any extra assumptions. We show that the proposed test is consistent. \\ 
	Test power and size are examined in Monte Carlo simulations in diverse location and scale problems under low sample size and compared to other parametric and interpoint distance based tests. In particular, its robustness to dimensionality has been evidenced by simulations making our test particularly suitable for high-dimensional settings compared to its competitors. As a proof of concept, we illustrate the feasibility of our test on a real data example.
	To this end, we show that the  results of our averaged Wilcoxon test are consistent with the results of a differential expression study of high-dimensional genomic microarray measurements. \\
	The outline of this article is as follows. In Section \ref{Sec:Motivation} we introduce and derive the necessary quantities for the implementation of the test, show the tests connection to the Wilcoxon rank sum test and motivate the variance approximation. The theoretical results, the estimate of the variance of the test statistics and its upper bound and asymptotic normality and are presented in Section \ref{Sec:Res}. The soundness of the assumptions of the theorems on the underlying distributions are illustrated by examples. The results of Monte Carlo simulations are given in Section \ref{Sec:Monte Carlo} and microarray data analyzed in Section \ref{Sec:RealData}. The caveats and open problems are then discussed in Section \ref{Sec:Conclusion}.\\
	
	\noindent
	\textit{Notation.} Metric (or a distance function) is denoted by $d$ while $D \in \N$ denotes the dimension of the underlying space. Multiindices are denoted with $\mathbf{i}$. If not otherwise specified, it is $\mathbf{i} \in \N^4$, i.e. $\mathbf{i}= (i_1, i_2, i_3, i_4)$. Given a multiindex $\mathbf{i}$, the quantity $\I_{\mathbf{i}}$ abbreviates $\I_{\{ d(X_{i_1}, X_{i_2}) \leq d(X_{i_3}, Y_{i_4}) \} }$.  Given $f,g \underset{x \rightarrow \infty}{ \rightarrow} \infty $, we write $f \sim g$ iff $ \frac{f}{g} \underset{x \rightarrow \infty}{ \rightarrow} C < \infty $ with $C>0$.
	
	\section{Mathematical Preliminaries}\label{Sec:Motivation}

	We are interested in the general multivariate two-sample homogeneity problem: 
	Given i.i.d.  observations $X_1, \ldots, X_n$ in $\mathbb{R}^D$ from a population with distribution function $F_X$  and i.i.d. observations $Y_{n+1}, \ldots, Y_{n+m}$ in $\mathbb{R}^D$  from a population with distribution function $F_Y$, our goal is to test
	\begin{equation} \label{eqn::testprob0}
		\mathcal{H}_0: F_X = F_Y \textnormal{    vs.   } \mathcal{H}_1: F_X\neq  F_Y, 
	\end{equation}
	\noindent i.e. we wish to infer whether the two samples are generated by the same probability distribution. For this purpose,  a distinguishing quantity is needed to base a test on. As motivated in the introduction, we will use  interpoint distances to transform  potentially high-dimensional i.i.d.\ data to a univariate sample of dependent interpoint distances.
	For a given distance function 
	$d: \R^D \times \R^D  \rightarrow \left[ 0, \infty \right) $ we distinguish within-sample interpoint distances, i.e. distances  $d(X_i, X_j)$, $1\leq i, j\leq n$,  or  $d(Y_i, Y_j)$, $n+1\leq i, j\leq n+m$, and between-samples interpoint distances, i.e.\  
	distances  $d(X_i, Y_j)$, $1\leq i\leq n$,
	$n+1\leq j\leq n+m$.
	\noindent
	Given two samples drawn from distinct distributions $F_X$ and $F_Y$, we expect the within-sample interpoint distances with respect to $F_X$, the within-sample interpoint distances with respect to $F_Y$, and the between-samples interpoint distances to follow distinct
	distributions. 
	In fact, Theorem 2 in \cite{Maa96} provides mild sufficient conditions on data distributions and distance functions for this to hold true. 
	More precisely, the theorem states the following:
	Let $X_1, X_2, X_3$  be i.i.d. random vectors with   Lebesgue probability density $f_X$, let $Y_1, Y_2, Y_3$ be i.i.d. random vectors with   Lebesgue probability density $f_Y$ and let $X_1, X_2,X_3$, $Y_1, Y_2,Y_3$ be independent. Moreover, assume that
	\begin{enumerate}[start=1,label={(\bfseries R\arabic*):}]
		\item $\int_{\mathbb{R}^D}f^2_X(x) dx, \int_{\mathbb{R}^D}f^2_Y(y)dy<\infty$,
		\item The zero vector is a Lebesgue point of the function $u(y)= \int f_Y(x+y) f_X(x) dx$, i.e. it holds that $\frac{1}{\lambda(B_r(0))} \int_{B_r(0)} |u(y) - u(0)| d y \underset{r \rightarrow 0}{\rightarrow} 0$, where $B_r(x)$ denotes the ball in $\mathbb{R}^d$ with radius $r$ around $x$ and $\lambda$ the Lebesgue measure on $\mathbb{R}^D$.
	\end{enumerate}
	Then, if
	\begin{enumerate}[start=1,label={(\bfseries D\arabic*):}]
		\item $d(x,y)= 0$ if, and only if,  $x=y$,
		\item for all $a \in \R$ and $x,y,b \in \R^D$ $d(ax+b, ay+b) = |a|d(x,y)$,
	\end{enumerate}
	it holds that
	\begin{equation*}
		f_X = f_Y \textit{\quad if, and only if,  \quad } d(X_1, X_2) \overset{\mathcal{D}}{=} d(Y_1, Y_2) \overset{\mathcal{D}}{=}d(X_3, Y_3).
	\end{equation*}
	According to this result, it suffices to identify a change in distribution from  within-sample interpoint distances to between-samples interpoint distances in order to reject the hypothesis $\mathcal{H}_0: F_X=F_Y$. In order to derive a test for \eqref{eqn::testprob0},
	we may therefore as well consider the test problem
	\begin{equation} \label{eqn::testprob}
		\mathcal{H}_0: d_X \overset{\mathcal{D}}{=}d_{XY}  \textnormal{    vs.   } \mathcal{H}_1: d_X \overset{\mathcal{D}}{\neq}d_{XY}, 
	\end{equation}
	where $d_X:=d(X_1, X_2)$ and $d_{XY}:=d(X_3, Y_1)$
	for i.i.d. random vectors $X_1, X_2, X_3$ with Lebesgue density $f_X$ and $Y_1$ a random vector with Lebesgue density $f_Y$.

	In order to obtain a robust test under minimal conditions, we aim to base a test decision  for the above testing problem on a version of  the two-sample Wilcoxon rank sum test, applied to  the   within-sample interpoint distances and the between-samples interpoint distances. To this end, we denote the given point clouds in $\R^D$ by $\mathcal{X}=(X_i)_{i = 1}^n$ and $\mathcal{Y}=(Y_j)_{j= n+1}^{n+m}$. Moreover, we introduce the index set
	\begin{align*}
		I_{\mathcal{X}, \mathcal{X}}:=\{(i_1, i_2) | 1\leq i_1 < i_2 \leq n\}
	\end{align*}
	to label all distances within $\mathcal{X}$ and the index set 
	\begin{equation*}
		\ I_{\mathcal{X}, \mathcal{Y}}:=\{(i_1, i_2) | i_1 \in \left[ n \right],  i_2 \in \left[ n + m \right] \setminus \left[ n \right] \}
	\end{equation*}
	to label all distances between points in $\mathcal{X}$ and points in $\mathcal{Y}$. For $\mathbf{i}=(i_1, i_2) \in  I_{\mathcal{X}, \mathcal{X}}$ and $ \mathbf{j}=(j_1, j_2) \in  I_{\mathcal{X}, \mathcal{Y}}$ we write $d_{\mathbf{i}} = d(X_{i_1}, X_{i_2}) $ and $d_{\mathbf{j}} = d(X_{j_1}, Y_{j_2}) $. 
	
	\subsection{The test statistic}
	
	Computing the Wilcoxon two-sample test statistic with respect to all interpoint distances 
	induces dependencies which make the  theoretical analysis of a corresponding hypothesis test particularly hard.
	Another approach to the testing problem   could instead be based on independent interpoint distances only. In this particular case, the original Wilcoxon test could be used without any adjustments. However, in terms of power, a test based on all distances will in general clearly outperform a test based on independent distances.
	To formalize our approach, which considers most but not all interpoint distances, we define 
	the set of all index sets corresponding to a set of  $r$ independent within-sample and between-samples distances by
	\begin{equation*}
		\mathcal{I}_r^{(n,m)}:=\left\{I\subset  I_{\mathcal{X}, \mathcal{X}} \cup  I_{\mathcal{X}, \mathcal{Y}}|
		\# \left(I\cap I_{\mathcal{X}, \mathcal{X}}\right)= \# \left(I\cap I_{\mathcal{X}, \mathcal{Y}}\right)=r, \textnormal{ such that } \forall \ \mathbf{i}, \mathbf{j}\in I \ \mathbf{i}\cap 
		\mathbf{j}=\emptyset\right\},
	\end{equation*}
	where, with a slight abuse of notation, $\mathbf{i}\cap \mathbf{j}=\{i_1, i_2\} \cap  \{j_1, j_2\}$. 
	To include as much information on the distribution of interpoint distances as possible in the test decision, 
	we aim at choosing  $r$ as large as possible. 
	A maximal choice of $r$ corresponds to $r = \lfloor\frac{n}{3}  \rfloor$. Each set $I \in \mathcal{I}_r^{(n,m)}$ corresponds to a set of distances 
	\begin{equation*}
		D(I):=\left\{d_{\mathbf{i}}|\mathbf{i}\in I\cap I_{\mathcal{X}, \mathcal{X}}\right\}\cup \left\{d_{\mathbf{j}}| \mathbf{j}\in I\cap I_{\mathcal{X}, \mathcal{Y}}\right\}.
	\end{equation*}
	Accordingly,  a choice of $I$ 
	corresponds to a set $D(I)$ of independent instances of $d_X$ and $d_{XY}$.
	The corresponding two-sample Wilcoxon statistic restricted to the set $I\in\mathcal{I}_r^{(n,m)}$ is given by
	\begin{equation}
		W^I =\sum\limits_{\mathbf{j}\in I\cap I_{\mathcal{X}, \mathcal{Y}}}R_{\mathbf{j}}^I
	\end{equation}
	where $R_{\mathbf{j}}^I$ is the rank of $d_{\mathbf{j}}$ in $D(I)$,i.e.
	\begin{equation}
		\ R_{\mathbf{j}}^I= \sum_{\mathbf{i} \in  I}\I_{ \{d_{\mathbf{i}} \leq d_{\mathbf{j}} \} }.
	\end{equation}

	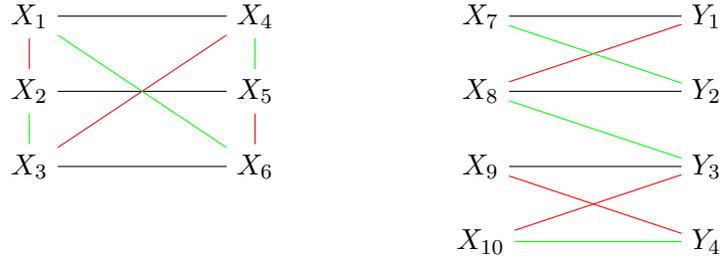
\begin{figure}
		\centering
		\begin{tikzpicture}
			
			\foreach \x [count=\i] in {1,2,3}
			\node (X\x) at (0, -\i) {$X_{\x}$};
			
			\foreach \x [count=\i] in {4,5,6}
			\node (X\x) at (3, -\i) {$X_{\x}$};
			
			\foreach \x [count=\i] in {7,8,9, 10}
			\node (X\x) at (6, -\i) {$X_{\x}$};
			
			\foreach \y [count=\i] in {1, ..., 4}
			\node (Y\y) at (9, -\i) {$Y_{\y}$};
			
			\usetikzlibrary {graphs}
			\graph {
				(X1) -- (X4), (X2) -- (X5), (X3) -- (X6), (X7) -- (Y1), (X8) -- (Y2), (X9) -- (Y3),
				(X1) --[color=red] (X2), (X5) --[color=red] (X6), (X3) --[color=red] (X4), (X9) --[color=red] (Y4),  (X8) --[color=red] (Y1),  (X10) --[color=red] (Y3),
				(X1) --[color=green] (X6), (X4) --[color=green] (X5), (X2) --[color=green] (X3), (X7) --[color=green] (Y2), (X10) --[color=green] (Y4), (X8) --[color=green] (Y3)
			};

		\end{tikzpicture}

		\caption{An example of three different independent subsamples $D_1, D_2, D_3\in \mathcal{D}$ for $n=10$ and $m=4$ depicted in red, green and black.}
		\label{fig:Subsampling}
	\end{figure}
	\noindent
	Separating randomness from the sampling of distances and the randomness of observations,  where we assume that all samplings described above are equally probable, we finally define the test statistic 
	\begin{equation} \label{eqn::original_T}
		\mathcal{T}^{(n,m)}:= \E \left[ W^I  | I \in \mathcal{I}_r^{(n,m)} \right].
	\end{equation}
	Since for any  $\mathbf{j} \in \left[ n \right] \times ( \left[ n+m \right] \setminus \left[ n \right])$
	\begin{equation*}
		\E \left[ R_{\mathbf{j}}^I| I \in \mathcal{I}_r^{(n, m)}  \right] = \frac{1}{\#\mathcal{I}_r^{(n, m)}}{\underset{I \in \mathcal{I}_r^{(n, m) } \textnormal{ s.t. } \mathbf{j} \in I}{\sum}}   R_{\mathbf{j}}^{I}, 
	\end{equation*}
	we have 
	\begin{align*}
		\mathcal{T}^{(n,m)} &= \frac{1}{\#\mathcal{I}_r^{(n, m)}} \sum_{I \in \mathcal{I}_r^{(n, m)}} \sum_{ \mathbf{j} \in I \cap I_{\mathcal{X}, \mathcal{Y}}} \sum_{\mathbf{i} \in I}\I_{ \{d_{\mathbf{i}} \leq d_{\mathbf{j}} \} } \\
		&=\lfloor\frac{n}{3}  \rfloor + \frac{1}{\#\mathcal{I}_r^{(n, m)}} \sum_{I \in \mathcal{I}_r^{(n, m)}} \sum_{ \mathbf{j} \in I \cap I_{\mathcal{X}, \mathcal{Y}}} \underset{\mathbf{j} \neq \mathbf{i}}{\sum_{\mathbf{i} \in I}}\I_{ \{d_{\mathbf{i}} \leq d_{\mathbf{j}} \} }.
	\end{align*}
	
	\noindent
	For $\mathbf{i}, \mathbf{j}  \in \left[ n\right] \times \left[ n + m\right]$, we define the set 
	\begin{equation}\label{def::set_d}
		\mathfrak{d}^{(n,m)}(\mathbf{i}, \mathbf{j}) :=  \{ I \in  \mathcal{I}_r^{(n, m)} |  \mathbf{i}, \mathbf{j} \in  I, \mathbf{j} \in  I_{\mathcal{X}, \mathcal{Y}}\}
	\end{equation}
	which collects all the distance samplings $I$ which  entail the comparison $d_{\mathbf{i}} \leq d_{\mathbf{j}}$. There are only three possibilities for the number of elements in $\mathfrak{d}^{(n,m)}$. The first is,  $d_{\mathbf{i}}$ and $d_{\mathbf{j}}$ are never compared, which is the case if the tupels $\mathbf{i}$ and $\mathbf{j}$ share an index. Then the set is empty and thus $\# \mathfrak{d}^{(n,m)}(\mathbf{i}, \mathbf{j}) = 0$.  Otherwise, if the set (\ref{def::set_d}) is non-empty, due to combinatorial symmetry, i.e. each comparison occurring equally often, its cardinality depends only on the type of comparison. We therefore define $\# \mathfrak{d}_1^{(n,m)} := \# \mathfrak{d}^{(n,m)}(\mathbf{i}, \mathbf{j})$ for $\mathbf{i}\in I_{\mathcal{X}, \mathcal{X}}$, $ \mathbf{j} \in I_{\mathcal{X}, \mathcal{Y}}$  and  $\# \mathfrak{d}_2^{(n,m)} := \# \mathfrak{d}^{(n,m)}(\mathbf{i}, \mathbf{j})$ for $\mathbf{i}, \mathbf{j} \in I_{\mathcal{X}, \mathcal{Y}}$. 
	\\
	\noindent
	To simplify notation, we define index sets $I_1^{(n,m)}$ and $I_2^{(n,m)}$ corresponding to all admissible within-$\mathcal{X}$  distance comparisons for our test as 
	\begin{equation}\label{eq:I1}
		I_1^{(n,m)} := \{ \mathbf{i}=(i_1, i_2, i_3, i_4) \in \N^4 | i_1, i_2, i_3 \in \left[n\right], i_4 \in \left[ n+m\right] \setminus \left[ n \right], \textnormal{ where } i_1 < i_2, i_1 \neq i_3, i_2 \neq i_3 \}
	\end{equation}
	and between-$\mathcal{X}$-and-$\mathcal{Y}$
	\begin{equation}
		I_2^{(n,m)} := \{ \mathbf{i}=(i_1, i_2, i_3, i_4)  \in \N^4 | i_1, i_3 \in \left[n\right], i_2, i_4 \in \left[ n+m\right] \setminus \left[ n \right], \textnormal{ where } i_1 \neq i_3, i_2 \neq i_4\}.
	\end{equation}
	We further consider $4$-multiindices, which we denote with $\mathbf{i}$. Moreover, in the following, we write $\I_{\mathbf{i}}:=  \I_{\{d(X_{i_1}, X_{i_2}) \leq d(X_{i_3}, Y_{i_4})\}}$ for $\mathbf{i} \in I_1^{(n,m)}$ and 
	$\I_{\mathbf{j}}:=  \I_{\{d(X_{j_1}, Y_{j_2}) \leq d(X_{j_3}, Y_{j_4})\}}$ for $\mathbf{j} \in I_2^{(n,m)}$. Accordingly, since we consider balanced index sets, we can rewrite $\mathcal{T}^{(n,m)}$ as the following rank sum statistic
	\begin{equation}\label{eqn::TT}
		\mathcal{T}^{(n,m)} = \lfloor\frac{n}{3}  \rfloor + C_1^{(n,m)}  \sum_{\mathbf{i} \in I_1^{(n,m)}} \I_{\mathbf{i}} +  C_2^{(n,m)}  \sum_{\mathbf{j} \in I_2^{(n,m)}} \I_{\mathbf{j}},
	\end{equation}
	where
	\begin{equation} \label{eqn::C1C2}
		C_1^{(n,m)} := \frac{\# \mathfrak{d}_1^{(n,m)} }{\# \mathcal{I}_r^{(n, m)}}  \quad \text{and} \quad  C_2^{(n,m)} := \frac{\# \mathfrak{d}_2^{(n,m)} }{\# \mathcal{I}_r^{(n, m)}}.
	\end{equation}
	The sizes of $ I_1^{(n,m)}$ and  $I_2^{(n,m)}$ and the corresponding constants  $C_1^{(n,m)}$ and $C_2^{(n,m)}$ are calculated in Section \ref{subsec::combinatorial_quantities}. \\
	
	\noindent
	\textit{Properties of the test statistic.} We further provide some computations that prepare our asymptotic considerations. We start by rewriting the test statistic $\mathcal{T}^{(n,m)}$ in terms of a centered sequence and introduce a graph encoding the  dependencies occurring when interpoint distances are compared. Based on this, we will subsequently compute the variance of the test statistic to prepare the derivation of  a distribution-free bound later on. 
	To this end, let $I^{(n,m)} := I_1^{(n,m)} \cup I_2^{(n,m)}$ and define
	$\mathbf{E}_1:= \Prob \left(d(X_1, X_2) \leq d(X_3, Y_1)\right)$ and $\mathbf{E}_2:= \Prob \left(d(X_1, Y_1) \leq d(X_2, Y_2)\right)$, where $X_1, X_2, X_3, Y_1, Y_2$ are independent random vectors,  $X_1, X_2, X_3$  with distribution function $F_X$ and  $Y_1, Y_2$  with distribution function $F_Y$. We define
	the triangular array $\mathcal{Z}=\{\mathcal{Z}_{\mathbf{i}}^{(n,m)}\,|\,\mathbf{i}\in I^{(n,m)}\}$, where the $\mathcal{Z}_{\mathbf{i}}^{(n,m)}$ are given by 
	\begin{equation} \label{eqn::triangle}
		\mathcal{Z}_{\mathbf{i}}^{(n,m)} := \begin{cases}
			C_1^{(n,m)}  ( \I_{\mathbf{i}} - \mathbf{E}_1)  &  \text{ for } \mathbf{i} \in I_1^{(n,m)} \\
			C_2^{(n,m)} (  \I_{\mathbf{i}} - \mathbf{E}_2 ) & \text{ for } \mathbf{i} \in I_2^{(n,m)} 
		\end{cases}
	\end{equation}
	such that 
	\begin{equation*}
		\mathcal{T}^{(n,m)} - \E  \mathcal{T}^{(n,m)} = \sum_{\mathbf{i} \in I^{(n,m)}} \mathcal{Z}_{\mathbf{i}}^{(n,m)}.
	\end{equation*}
	The random variables in  $\mathcal{Z}$ are  dependent. This is because the sets $I_1^{(n,m)}$ and $I_2^{(n,m)}$ contain almost all possible interpoint distance comparisons. In the following, we will describe these dependencies by means of a \textit{graph}. \\
	\noindent
	The relation in our case is the pairwise dependence of random vectors. We define the graph of dependencies within the array $(\mathcal{Z}_{\mathbf{i}}^{(n,m)})_{\mathbf{i} \in I^{(n,m)}}$ as $G^{(n,m)} = (I^{(n,m)}, E^{(n,m)})$, where the relation $E^{(n,m)}$ is defined by 
	\begin{equation*}
		(\mathbf{i}, \mathbf{j}) \in E^{(n,m)} :\iff \left(\mathbf{i} \cap \mathbf{j} \neq \emptyset \ \wedge \ \mathbf{i} \neq \mathbf{j}\right)
	\end{equation*}
	where, again with a slight abuse of notation, $\mathbf{i}\cap \mathbf{j}=\{i_1, i_2, i_3, i_4\} \cap  \{j_1, j_2, j_3, j_4\}$. This graph, which is also referred to as \textit{dependency graph} in the literature \cite{Ross2011, Austern22}, is  undirected since $ (\mathbf{i}, \mathbf{j}) \in E^{(n,m)} $ is equivalent to   $(\mathbf{j}, \mathbf{i}) \in E^{(n,m)}$. Moreover, it does not allow for self-loops, i.e. $(\mathbf{i},\mathbf{i}) \notin E^{(n,m)}$. 
	The variance of $\mathcal{T}^{(n,m)}$ can then be expressed as
	\begin{align*}
		\Var \left(\mathcal{T}^{(n,m)}\right) &= \sum_{\mathbf{i}, \mathbf{j} \in I^{(n,m)}} \E \left[ \mathcal{Z}_{\mathbf{i}}^{(n,m)} \mathcal{Z}_{\mathbf{j}}^{(n,m)} \right] \\
		&= \sum_{\mathbf{i} \in I^{(n,m)}} \E \left[ (\mathcal{Z}_{\mathbf{i}}^{(n,m)})^2 \right] + \sum_{(\mathbf{i}, \mathbf{j}) \in E^{(n,m)}} \E \left[ \mathcal{Z}_{\mathbf{i}}^{(n,m)} \mathcal{Z}_{\mathbf{j}}^{(n,m)} \right].
	\end{align*}
	The previous preliminary computations now allow us to derive a more explicit  representation for the variance of our test statistic under the null hypothesis.
	\begin{lem}
		Under the null hypothesis, we have 
		\begin{align*}
			\Var\left(\mathcal{T}^{(n,m)}\right) =& (\mathbf{E}_1-\mathbf{E}_1^2)\left[ \# I_1^{(n,m)} (C_1^{(n,m)})^2 +  \# I_2^{(n,m)} (C_2^{(n,m)})^2 \right] \\
			&+ \left[(1-\mathbf{E}_1^2) \cdot \Prob(\I_{\mathbf{i}} = \I_{\mathbf{j}}=1)+ 
			\mathbf{E}_1^2\cdot\Prob(\I_{\mathbf{i}} = \I_{\mathbf{j}}=0)
			-\mathbf{E}_1(1-\mathbf{E}_1)\cdot\Prob( \I_{\mathbf{i}} \neq \I_{\mathbf{j}}) \right] \\
			&\qquad\qquad\times\begin{cases}
				(C_1^{(n,m)})^2 & \text{ for } \ \mathbf{i}, \mathbf{j}  \in I_1^{(n,m)} \\
				(C_2^{(n,m)})^2 & \text{ for } \ \mathbf{i}, \mathbf{j}  \in I_2^{(n,m)} \\
				C_1^{(n,m)}C_2^{(n,m)} & \text{ for } \ \mathbf{i}  \in I_1^{(n,m)},  \ \mathbf{j} \in  I_2^{(n,m)}.
			\end{cases}
		\end{align*}
	\end{lem}
	\begin{proof}
		Under $\mathcal{H}_0$ we have  $\mathbf{E}_1=\mathbf{E}_2$.
		It therefore holds that
		\begin{equation*}
			\E \left[ (\mathcal{Z}_{\mathbf{i}}^{(n,m)})^2 \right] = \Var\left(\I_{\mathbf{i}}\right) \begin{cases}
				(C_1^{(n,m)})^2 & \text{ for } \mathbf{i} \in I_1^{(n,m)} \\
				(C_2^{(n,m)})^2 & \text{ for } \mathbf{i} \in I_2^{(n,m)} 
			\end{cases}
			= \mathbf{E}_1(1-\mathbf{E}_1) \begin{cases}
				(C_1^{(n,m)})^2 & \text{ for } \mathbf{i} \in I_1^{(n,m)} \\
				(C_2^{(n,m)})^2 & \text{ for } \mathbf{i} \in I_2^{(n,m)} 
			\end{cases}
		\end{equation*}
		and
		\begin{equation*}
			\E \left[ \mathcal{Z}_{\mathbf{i}}^{(n,m)} \mathcal{Z}_{\mathbf{j}}^{(n,m)} \right] = \E \left[ \I_{\mathbf{i}} \I_{\mathbf{j}} -\mathbf{E}_1(\I_{\mathbf{i}} +\I_{\mathbf{j}}) + \mathbf{E}_1^2 \right]  \begin{cases}
				(C_1^{(n,m)})^2 & \text{ for } \mathbf{i}, \mathbf{j}  \in I_1^{(n,m)} \\
				(C_2^{(n,m)})^2 & \text{ for } \mathbf{i}, \mathbf{j}  \in I_2^{(n,m)} \\
				C_1^{(n,m)}C_2^{(n,m)} & \text{ for } \mathbf{i}  \in I_1^{(n,m)},  \mathbf{j} \in  I_2^{(n,m)}.
			\end{cases}
		\end{equation*}
		From this, it follows that 
		\begin{equation}\label{eq:var1}
			\I_{\mathbf{i}} \I_{\mathbf{j}} -\mathbf{E}_1(\I_{\mathbf{i}} +\I_{\mathbf{j}}) + \mathbf{E}_1^2 =  \begin{cases}
				(1-\mathbf{E}_1)^2 & \text{ for } \I_{\mathbf{i}} = \I_{\mathbf{j}}=1 \\
				\mathbf{E}_1^2& \text{ for } \I_{\mathbf{i}} = \I_{\mathbf{j}}=0 \\
				-\mathbf{E}_1(1-\mathbf{E}_1) & \text{ for } \I_{\mathbf{i}} \neq \I_{\mathbf{j}}, 
			\end{cases}
		\end{equation}
		which implies the claim and hence finishes this proof.
	\end{proof}
	
	\noindent
	\textit{$\Var(\mathcal{T}^{(n,m)})$ depends on the underlying distributions.} Note that, in general, the quantity $\Prob (\I_{\mathbf{i}} = \I_{\mathbf{j}})$ is not  distribution-free. We will consider this probability in the following in a particular simple case, namely under the null hypothesis, under the assumption that the  underlying distribution  does not produce ties, and  given that $\mathbf{i}$ and $\mathbf{j}$ share exactly one index, e.g. $\mathbf{i}=(1, 2, 3, 4)$ and $\mathbf{j}=(1,5, 6, 7)$ such that   $\Prob (\I_{\mathbf{i}} = \I_{\mathbf{j}})=\Prob (\I_{\{d(X_1,X_2)\leq d(X_3, X_4)\}} = \I_{\{d(X_1,X_5)\leq d(X_6, X_7)\}})$. Under the null hypothesis, this probability depends on the distribution of $X$, the  dimension $D$ and the choice of distance $d$ as evidenced by numerical simulations; see Tables \ref{tab::Normal}-\ref{tab::Pareto}. We further compute $\Prob (\I_{\mathbf{i}} = \I_{\mathbf{j}})$ for specific examples. To this end, we use the following result.
	\begin{lem} \label{lemma::example1}
		Let  $X, X_1, \ldots, X_7 \in \R^D$ be i.i.d  random vectors, such that 
		\begin{equation*}
			\Prob (d(X_1,X_2) = d(X_3, X_4)) = 0,
		\end{equation*}
		i.e. ties do not occur a.s. Let
		\begin{equation} \label{eqn::P1}
			P_1 := \Prob (\I_{\{d(X,X_2)\leq d(X_3, X_4)\}} = \I_{\{d(X,X_5)\leq d(X_6, X_7)\}}). 
		\end{equation}
		Assume further that the metric $d$ is induced by a norm $\|.\|$ on $\R^D$. Then, it holds that
		\begin{equation} 
			P_1 = \E_X \left[ 1- 2p(X) + 2p^2(X)\right]
		\end{equation}
		with $p(x)= \Prob( X_2 \in \mathcal{B}_{\|X_4 - X_3\|}(x))$, where $\mathcal{B}_{\|X_4 - X_3\|}(x)$ denotes the  ball with center $x$ and  (random) radius $\|X_4 - X_3\|$.
	\end{lem}
	\begin{proof}
		Let
		\begin{equation*}
			f(X, Y) :=  \I_{\{\I_{\{d(X,X_2)\leq d(X_3, X_4)\}} = \I_{\{d(X,X_5)\leq d(X_6, X_7)\}} \}},
		\end{equation*}
		where $Y := (X_2,...,X_7)$. It then follows that 
		\begin{equation*}
			P_1 = \E \left[ f(X,Y) \right] = \E \left[ \E_Y\left[ f(X,Y) | X \right] \right] = \E_X \left[ h(X) \right],
		\end{equation*}
		where $h(x) := \E_Y \left[ f(x,Y) \right] $. On the other hand, we have
		\begin{equation*}
			h(x)= \Prob(Z_1(x) = Z_2(x)),
		\end{equation*}
		where $Z_1(x):= \I_{\{d(x,X_2)\leq d(X_3, X_4)\}} $ and $Z_2(x):= \I_{\{d(x,X_5)\leq d(X_6, X_7)\}}$. Note that  $Z_1$ and $Z_2$ are  two independent Bernoulli distributed random variables and hence
		\begin{equation*}
			h(x)= 1-2p(x)+2p^2(x)
		\end{equation*}
		with $p(x):= \Prob(Z_1(x) = 1) = \E Z_1(x)$ i.e.  
		\begin{equation} \label{eqn::P1_p}
			P_1 = \E_X \left[ 1- 2p(X) + 2p^2(X)\right].
		\end{equation}
		Under the assumption that the distance function $d$ is induced by a norm $\| \cdot\|$, it holds that 
		\begin{equation*}
			p(x)= \E \left(\I_{ \{ \|X_2-x\| \leq \|X_4 - X_3\| \} }\right)=\Prob\left( X_2 \in \mathcal{B}_{\|X_4 - X_3\|}(x)\right).
		\end{equation*}
	\end{proof}
	
	\noindent
	We consider specific examples for which the dimension $D$ equals one (i.e. $X,  X_1, \ldots, X_4$ take values in $\R$) and $\|\cdot\| = |\cdot|$.\\
	\begin{bsp} \label{bsp::e1}
		Suppose that $X$ has a density  $f$. Then, in order to compute $P_1$, note that  $p(x)= \E_{X_3}\left(\E_{X_2, X_4} \left(\I_{ \{ \|X_2-x\| \leq \|X_4 - X_3\| \} }\right) \right)$ with
		\begin{equation*}
			\E_{X_2, X_4} \left(\I_{ \{ \|X_2-x\| \leq \|X_4 - X_3\| \} }\right)= \int\int \I_{ \{ \|x_2\| \leq \|x_4\|\} } f(x_2 + x) f(x_4 + X_3) dx_2 dx_4.
		\end{equation*}
		\begin{enumerate}
			\item Let $X, X_1, \ldots ,X_4$ be uniformly distributed on $[0, 1]$, i.e. $f(x)=\I_{\left[0,1 \right]}(x)$. In this case,  $\E\left( X^n \right)= \frac{1}{n+1}$. Straightforward calculations then yield
			\begin{equation*}
				p(x)= \left(\frac{1}{3} +x-x^2\right) \I_{\left[0,1\right]}(x)
			\end{equation*}
			and therefore
			\begin{equation*}
				P_1 = \E\left[ \frac{5}{9} - \frac{2}{3}x + \frac{8}{3}x^2-4x^3 +2x^4 \right] \approx 0.51.
			\end{equation*}
			\item Let $X, X_1, \ldots , X_4$ be exponentially distributed with parameter $\lambda$, i.e.   $f(x) = \lambda e^{-\lambda x} \I_{(0, \infty)} (x)$. `It then follows that
			\begin{equation*}
				p(x)= \lambda x e^{-\lambda x} +\frac{1}{2} e^{-3\lambda x} + \frac{1}{4} e^{-4\lambda x} -\frac{1}{4} e^{-3\lambda x}
			\end{equation*}
			and straightforward calculations yield $P_1 \approx 0.55$.
		\end{enumerate}
	\end{bsp}
	\noindent
	The explicit calculations of $P_1$ in the above examples are consistent with the numerically estimated values; see Tables \ref{tab:: Uniform} and \ref{tab::Gamma}. It will be established in the proof of Theorem \ref{thm::Var} that  the case where $\mathbf{i}$ and $\mathbf{j}$ share exactly one index asymptotically contributes most to the variance of $\mathcal{T}^{(n,m)}$. Therefore, its estimation plays a major role in the design of our testing procedure  and  the computation of critical values.\\

	\noindent
	\textit{Variance of $\mathcal{T}^{(n,m)}$ if ties are excluded. } Hereafter we limit ourselves to such distributions and metrics, which do not produce ties. In this case, the variance of $\mathcal{T}^{(n,m)}$ has a simpler form.
	\begin{lem}\label{lem:var}
		Under $\mathcal{H}_0$, if ties in distance comparisons do not occur almost surely, i.e.
		\begin{equation*}
			\Prob (d(X_1,X_2) = d(X_3, X_4)) = 0,
		\end{equation*}
		it holds that 
		\begin{align*}
			\Var\left(\mathcal{T}^{(n,m)}\right) &=  \frac{1}{4} (C_1^{(n,m)})^2 \sum_{(\mathbf{i}, \mathbf{j}) \in E_1^{(n,m)}} (2\Prob (\I_{\mathbf{i}} = \I_{\mathbf{j}}) - 1) +\frac{1}{4} (C_1^{(n,m)})^2 \# I_1^{(n,m)},
		\end{align*}
		where $E_1^{(n,m)}$ is the set of edges between elements of $I_1^{(n,m)}$.
	\end{lem}
	
	\begin{proof} Under the null hypothesis, when ties  in distance comparisons do not occur almost surely, 
		we have $\mathbf{E}_1=\mathbf{E}_2=\frac{1}{2}$ and equation \eqref{eq:var1} yields
		\begin{equation*}
			\I_{\mathbf{i}}\I_{\mathbf{j}} - \frac{1}{2}(\I_{\mathbf{i}} +\I_{\mathbf{j}}) + \frac{1}{4} =  \begin{cases}
				\frac{1}{4} & \text{ for } \I_{\mathbf{i}} = \I_{\mathbf{j}} \\
				-\frac{1}{4} & \text{ for } \I_{\mathbf{i}} \neq \I_{\mathbf{j}}.
			\end{cases}
		\end{equation*}
		As a result, we get
		\begin{align*}
			\E \left[ \I_{\mathbf{i}}\I_{\mathbf{j}} - \frac{1}{2}(\I_{\mathbf{i}} +\I_{\mathbf{j}}) + \frac{1}{4} \right] &= \frac{1}{4}(\Prob (\I_{\mathbf{i}} = \I_{\mathbf{j}}) - \Prob (\I_{\mathbf{i}} \neq \I_{\mathbf{j}})) \\
			&= \frac{1}{4}(2\Prob (\I_{\mathbf{i}} = \I_{\mathbf{j}}) - 1).
		\end{align*}
		Since for  $\mathbf{j}=(j_1, j_2, j_3, j_4)\in I_2^{(n, m)}$ and $\mathbf{j}'=(j_3, j_4, j_1, j_2)$, it holds that $\mathbf{j}'\in I_2^{(n, m)}$ and $\I_{\mathbf{j}}+\I_{\mathbf{j'}}=1$, the second sum in (\ref{eqn::TT}) is almost surely constant and the formula for     $\Var\left(\mathcal{T}^{(n,m)}\right)$ follows. 
	\end{proof}
	
	\noindent
	Most notably, Lemma  \ref{lem:var} establishes an expression for the variance of the statistic $\mathcal{T}^{(n,m)}$ that, in contrast to the expression for the variances in \eqref{eqn::TT}, does not depend on $C_2^{(n,m)}$ or $ I_2^{(n,m)}$.
	As a result, these quantities do not play a role in any further analysis of $\mathcal{T}^{(n,m)}$.
	In the following, we therefore  set $I^{(n,m)} := I_1^{(n,m)}$, $E^{(n,m)} := E_1^{(n,m)}$ on $I_1^{(n,m)}$ and we define the \textit{tie-free version} $T^{(n,m)}$ of the test statistic $\mathcal{T}^{(n,m)}$ as
	\begin{equation}\label{eqn::Tfinal}
		T^{(n,m)} :=  \sum_{\mathbf{i} \in I^{(n,m)}} \I_{\{d(X_{i_1}, X_{i_2}) \leq d(X_{i_3}, Y_{i_4})\}}
	\end{equation}
	for which 
	\begin{equation}
		\E  T^{(n,m)} =  \mathbf{E}_1 \# I^{(n,m)}.
	\end{equation}
	Following the previous lemma, under the null hypothesis we have
	\begin{equation} \label{eqn::VarT(n,m)}
		\Var (T^{(n,m)}) =  \frac{1}{4}  \sum_{(\mathbf{i}, \mathbf{j}) \in E^{(n,m)}} (2\Prob (\I_{\mathbf{i}} = \I_{\mathbf{j}}) - 1) +\frac{1}{4}  \# I^{(n,m)}.  \\
	\end{equation}

	\section{Main results}\label{Sec:Res}
	We consider test statistic (\ref{eqn::Tfinal}) for the test problem 
	\begin{equation*}
		\mathcal{H}_0: d_X \overset{\mathcal{D}}{=} d_{XY} \textnormal{ vs. }    \mathcal{H}_1: d_X \overset{\mathcal{D}}{\neq} d_{XY}. 
	\end{equation*}
	Firstly, we have almost sure convergence of the $ (\# I^{(n,m)})^{-1}T^{(n,m)}$ to its expected value.  In particular, under the null hypothesis, this expected value is $\frac{1}{2}$.
	\begin{Theorem} \label{thm::Tconv} 
		Let there be a constant $\lambda \in (0,1)$ such that $\frac{n}{n+m} \rightarrow \lambda$ as $n,m \rightarrow \infty$. Then, it holds that
		\begin{equation*}
			(\# I^{(n,m)})^{-1}T^{(n,m)} -\mathbf{E}_1 \overset{ a.s. }{\underset{n,m \rightarrow \infty}{ \longrightarrow}} 0
		\end{equation*}
		with $\mathbf{E}_1 = \Prob (d(X_1,X_2) \leq d(X_3, Y))$, where $X_1$, $X_2$ and $X_3$ are independent copies  of  $X$ and $X$ is a random vector with distribution function $F_X$.
	\end{Theorem}
	\noindent
	For the implementation of the test, preferably a non-data dependent and distribution-free estimate of the variance is needed. The following theorem characterizes the asymptotic behavior of the variance of the test statistic $T^{(n,m)}$.
	\begin{Theorem} \label{thm::Var}
		Let $\frac{n}{n+m} \underset{n,m \rightarrow \infty}{ \rightarrow} \lambda \in (0,1)$ and let $$P_1 := \Prob (\I_{\{d(X_1,X_2)\leq d(X_3, X_4)\}} = \I_{\{d(X_1,X_5)\leq d(X_6, X_7)\}}),$$ where  $(X_i)_{i \in  \left[ n \right] }$ and $(Y_i)_{i \in  \left[ m \right] }$ are two sequences of i.i.d. random vectors. Assume that the choice of metric $d$ is such that there are no ties in distance comparisons, i.e.  
		\begin{equation}\label{cnd::noties}
			\Prob(d(X_1, X_2) = d(X_3, Y_1)) = 0.
		\end{equation}
		Assume that $P_1>\frac{1}{2}$ and let $I^{(n,m)}$ be defined as in \eqref{eq:I1}.
		Under the null hypothesis, we then have 
		\begin{align}\label{eqn::VarTest}
			\Var_H (T^{(n,m)}):= \Var (T^{(n,m)}) = \frac{3}{16}\frac{n+m-6}{m-1} \# C_{i_1 = j_1} (2P_1-1) + O(n^6),
		\end{align}
		where $\# C_{i_1 = j_1} = \frac{1}{3} nm(m-1)(n-1)(n-2)(n-3)(n-4) $. \\
		\noindent
		In particular, for $n=m$ and $n,m \rightarrow \infty$ it holds that $\Var_H (T^{(n,m)}) = \frac{3}{8} \# C_{i_1 = j_1} (2P_1-1) + O(n^6)$.
	\end{Theorem}
	\noindent
	The order of the main term of the variance, i.e. the first summand on the right-hand side of \ref{eqn::VarTest}, 
	is determined by $\# C_{i_1 = j_1}$.
	Since by assumption $\frac{n}{n+m} \rightarrow \lambda \in (0,1)$ as $n, m\rightarrow \infty$, it follows that $n\sim m$ and, as a consequence thereof, that 
	$\# C_{i_1 = j_1}\sim n^7$.
	Asymptotically, the first summand dominates  the expression on the right-hand side of \ref{eqn::VarTest}, such that estimation of the variance by this summand can be considered consistent.
	Notably, simulations suggest that the remaining terms (denoted as $O(n^6)$ in \eqref{eqn::VarTest})
	is positive. Accordingly, and particularly for small sample sizes $n$ and $m$, the expression $\frac{3}{16}\frac{n+m-6}{m-1} \# C_{i_1 = j_1} (2P_1-1)$ tends to underestimate $\Var_H (T^{(n,m)})$.
	In any case, a corresponding variance estimation
	depends on the  unknown quantity 
	$P_1$.
	Choosing the trivial upper bound $1$ for  $P_1$, a corresponding test  has nontrivial power in some cases as  evidenced by  numerical simulations (cf. Section \ref{Sec:Monte Carlo}). 
	Nonetheless, to improve the estimation of the variance we establish non-trival  upper and lower bounds for $P_1$ through the following theorem.\\
	\begin{Theorem}\label{thm::P2}
		Let $P_1 := \Prob (\I_{\{d(X_1,X_2)\leq d(X_3, X_4)\}} = \I_{\{d(X_1,X_5)\leq d(X_6, X_7)\}})$, where $X_1, ...,X_7$ are  i.i.d. continuous  random vectors.
		Let further $S_X \subset \R^D$ be the support of the distribution  of $X$ and $S_R$ the support of the distribution  $\mathbb{P}_R$ of the  random variable $R=d(X_1,X_2)$, such that either
		\begin{enumerate}
			\item $S_X$ is unbounded and the metric $d$ on $\R^D$ is induced by a norm, or
			\item the distribution of $X$ satisfies 
			\begin{align}  \label{eqn::conc}
				\exists m, x \in S_X \textnormal{ such that   } &p_r(m) \geq p_r(x)  \textnormal{  } \forall  r \in S_R,\quad\text{and}\notag\\
				& p_r(m) > p_r(x)  \textnormal{  } \forall  r \in S_{0,R}\subset S_R\quad \text{with}\quad \mathbb{P}_R(S_{0,R})>0,
			\end{align}
			where $p_r(x):= \Prob (X \in \mathcal{B}_r(x))$ and $\mathcal{B}_r(x):=\{y\in \mathbb{R}^D:d(x, y)<r\}$ denotes the ball in $\mathbb{R}^D$ with radius $r$ and center $x$ with respect to the metric $d$. 
		\end{enumerate}
		Then, it holds that
		\begin{equation}\label{eqn::P1bound}
			\frac{1}{2} < P_1   < \frac{2}{3}.
		\end{equation}
	\end{Theorem}
	
	\begin{bem} \label{bem::thm2} Regarding condition \eqref{eqn::conc} we note the following.
		\begin{enumerate}
			\item 
			The function $p_r(\cdot)$ is closely related to  the pseudo-distance $$\delta_{\mathbb{P},\nu}:x\mapsto\inf\left\{r>0\,|\,\mathbb{P}\left(\bar{\mathcal{B}}_r(x)\right)>\nu\right\},$$
			defined in \cite{chazal2011geometric} in the context of geometric inference.
			\item Many standard bounded distributions fulfill  condition \eqref{eqn::conc}. 
				\end{enumerate}
		\begin{minipage}{0.38\textwidth}
			\hspace{0.4cm}
			\includegraphics[width=0.75\textwidth]{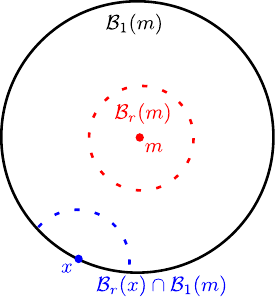}  
		\end{minipage}
			\begin{minipage}{0.6\textwidth}
				As a specific example consider  the uniform distribution on 
				$\mathcal{B}_1(0)$ and the Euclidean distance. Let $m = 0$ and let $x \in \mathcal{B}_1(0)$, $x\neq m$ . Then, it holds that
				\begin{align*}
					\begin{cases}
						&p_r(m)= p_r(x) \textnormal{  for  }  r \in \left[0, 1- \|m-x \| \right] \\
						&p_r(m)> p_r(x) \textnormal{  for  }  r \in \left(1- \|m-x\|,1 + \|m-x\| \right] \\
						&p_r(m)= p_r(x) \textnormal{  for  }  r > 1 + \|m-x\|,
					\end{cases}
				\end{align*}
				where $\|\cdot\|$ denotes the $l^2$-norm.
				In particular, for $x \in  \partial \mathcal{B}_1(0)$   condition \eqref{eqn::conc} holds with $S_{R,0}= S_R$ (cf.\ the sketch on the left hand side).
			\end{minipage}

	\end{bem}
	\noindent
	The conditions of Theorem \ref{thm::P2} are sufficient for the lower bound to hold. The upper bound holds for virtually any distribution and distance function when ties in distance comparisons do not occur (cf.\ the proof of Theorem \ref{thm::P2}). Moreover, note that the assumptions of Theorem \ref{thm::P2} guarantee that the first summand in \eqref{eqn::VarTest} does not vanish. As demonstrated by numerical experiments, for many distributions  the value of $P_1$ is much smaller than $\frac{2}{3}$ and, furthermore, depends on the dimension $D$ of the observations; see Tables \ref{tab::Normal}-\ref{tab::Cauchy}. A sharper, distribution-specific or general,  upper bound, if it exists, is therefore of interest. The natural question to examine is then if examples of distributions exist for which $P_1$ is close to $2/3$.\\
	\noindent
	For the proposed test statistic \eqref{eqn::Tfinal}, we establish asymptotic normality under  conditions on the distributions of $X$ and $Y$  which guarantee that the first summand in the variance of the statistic  $T^{(n,m)}$ as specified in \eqref{eqn::VarTest} does not vanish. In particular and notably, thereby the asymptotic distribution is not only established under the hypothesis that $X$ and $Y$ are equal in distribution, but also under more general assumptions.
	\begin{Theorem} \label{thm::T}
		Let $X:=(X_i)_{i\in [n]}$ and $Y:=(Y_i)_{i\in [m]}$ be sequences of independent, identically distributed random vectors.
		Assume that $m=m_n$ increases monotonically with $n$ such that
		$\frac{n}{n+m_n} \rightarrow \lambda \in (0,1)$ as $n\rightarrow \infty$.
		Moreover, assume that  the distance function $d$ is symmetric, that $r(Y)$, where  $r(y):=\Prob \left(d(X_1, X_2)\leq d(X_3, y)\right)$, is non-degenerate,  and that ties in distance comparisons do not occur, i.e. 
		$\Prob\left(d(X_1, X_2)= d(X_3, Y_1)\right)=0$.
		\ignore{the marginal distributions of $X$ and $Y$
			guarantee non-degeneracy of the random variable $p(X_1)$, where
			$p(x):= \Prob (X_1 \in \mathcal{B}_{d(X_2, Y_1)}(x))$, i.e. $p(X_1)$ is not almost surely constant.\\}
		\noindent
		Let 
		\begin{align} \label{eqn::Wnm}
			&W_{n,m} := \frac{T^{(n,m)} - \E T^{(n,m)}}{\sqrt{\Var (T^{(n,m)})}}, 
			\intertext{where}
			&T^{(n,m)} =  \sum_{\mathbf{i} \in I^{(n,m)}} \I_{\{d(X_{i_1}, X_{i_2}) \leq d(X_{i_3}, Y_{i_4})\}}, 
		\end{align}
		and $I^{(n,m)} := \{ \mathbf{i}=(i_1, i_2, i_3, i_4) \in \N^4 | i_1, i_2, i_3 \in \left[n\right], i_4 \in \left[ n+m\right] \setminus \left[ n \right],  i_1 < i_2, i_1 \neq i_3, i_2 \neq i_3 \}$.\\
		Then, it holds that
		\begin{equation} \label{eqn::asymp_norm}
			\mathcal{W}_1(\mathcal{L}(W_{n,m}), \mathcal{N}(0,1)) = \mathcal{O}(n^{-1/2}), 
		\end{equation}
		where $\mathcal{W}_1$ denotes the Wasserstein 1-distance.
	\end{Theorem}
	
	\begin{bem} \label{bem::thm3_bemerkung}
		\ignore{Note that the assumptions on $d$ and the marginal distributions of $X$ and $Y$ guaranteeing non-degeneracy of the random variable
			$p(X_1)$ are fulfilled for any semimetric $d$ and marginal distributions of $X$ and $Y$ with disjoint supports.
			Accordingly, the theorem holds under very general assumptions.
			Moreover, }Note that the assumption  that $r(Y)$ is non-degenerate is not only non-restrictive, but also necessary since otherwise $W_{n, m}$ is almost surely constant. 
		\ignore{$p(X)$ 
			almost surely constant implies $p(X) = \mathbf{E}_1$. Consider $\mathbf{E}_1 \in \{0,1\}$. In this case, $\E_X p(X) \in \{0,1\}$ and $p(x)\in\{0,1\}$ for $x \in supp f_X$, where $f_X$ is probability density of $X$. We conclude that $p(X)=\mathbf{E}_1$ is a.s. constant. In fact, in these cases it is $T^{(n,m)} \in \{ 0, \frac{1}{2} \# I^{(n,m)}  \}$ and $T^{(n,m)}$ a.s. constant. \\
			The condition, $p(X)$ a.s. not constant can for a given alternative be checked by examining the function $p_r(x) = \Prob( X \in \mathcal{B}_{r} (x))$.
			As an example, consider  $X_1 \sim U\left[0,1\right]$ and $Y_1 \sim U\left[a,a+1\right]$ where $a \in (0,1)$. In this case, we have 
			\begin{align*}
				p_r(a) & > p_r(1) \textnormal{ \quad for \quad} r \in [ 0,1-a) \\
				p_r(a) & \geq p_r(1)  \textnormal{ \quad for \quad} r \in [ 1-a, 1+a ] 
			\end{align*}
			and therefore 
			\begin{align*}
				p(a) &= \int_{X_1} \int_{|X-Y|} p_{|X-Y|}(a) f_{|X-Y|} f_{X_1} d |X-Y| d X_1  \\
				&  > \int_{X_1} \int_{|X-Y|} p_{|X-Y|}(1) f_{|X-Y|} f_{X_1} d |X-Y| d X_1 \\
				&= p(1)
			\end{align*}
			where $f_{|X-Y|}$ and $f_{X_1}$ are densities of $|X-Y|$ and $X$, therefore having $p(a)  > p(1)$ and $p(X)$ not a.s. constant as $p(x)$ is a continuous function. }
	\end{bem}
	\noindent
	According to Theorem \ref{thm::T} asymptotic normality of $W_{n, m}$ 
	holds for a suitable distance function $d$ and any two non-degenerate sequences $X:=(X_i)_{i\in [n]}$ and $Y:=(Y_i)_{i\in [m]}$ of independent, identically distributed random vectors
	with $P\left(d(X_1, X_2)\leq d(X_3, Y_1)\right)\in (0, 1)$.
	Consequently, under the null hypothesis, i.e. under the assumption that the marginal distributions of $X$ and $Y$ coincide, asymptotic  normality can be established under the general assumption of non-degenerate data-generating random vectors.  \\
	\begin{cor} \label{cor::folgerung1}
		Assume that $m=m_n$ increases monotonically with $n$ such that
		$\frac{n}{n+m_n} \rightarrow \lambda \in (0,1)$ as $n\rightarrow \infty$.
		Moreover, assume that  the distance function $d$ is symmetric, that ties in distance comparisons do not occur, i.e. 
		$\Prob\left(d(X_1, X_2)= d(X_3, X_4)\right)=0$ and that $r(X)$, where  $r(x):=\Prob \left(d(X_1, X_2)\leq d(X_3, x)\right)$, is a non-degenerate random variable.
		Then,   under $\mathcal{H}_0$, it holds that 
		\begin{equation*} 
			\mathcal{W}_1(\mathcal{L}(W_{n,m}), \mathcal{N}(0,1)) = \mathcal{O}(n^{-1/2}). 
		\end{equation*}
		
	\end{cor}
	\begin{bem} \label{bem::cor_bemerkung}
		Under the null hypothesis, since $Y$ has the same distribution as $X$, we have
		\begin{align*}
			\Var (r(X)) &= \frac{1}{2} \left(P_1 + 2 \E_X p(X) - 2(\E_X p(X))^2  -1\right) \\ 
			&=\frac{1}{2} \left(P_1 -\frac{1}{2}\right)
		\end{align*}
		due to Lemma \ref{lemma::example1}. 
		Under any of the conditions of Theorem \ref{thm::P2}, we have $P_1 > \frac{1}{2}$ and the conditions of Corollary \ref{cor::folgerung1} are fulfilled.
	\end{bem}
	\noindent
	Lastly, following  Theorems \ref{thm::Tconv}, \ref{thm::Var} and \ref{thm::T}, we conclude that the proposed test is consistent under a fixed alternative.\\
	
	\begin{Theorem} \label{thm::consistency}
		Let $n$ and $m$ be such that
		$\frac{n}{n+m} \rightarrow \lambda \in (0,1)$ as $n, m\rightarrow \infty$.
		Under any alternative for which $\mathbf{E}_1 \neq \frac{1}{2}$ and for which $r(X)$, where  $r(x):=\Prob \left(d(X_1, X_2)\leq d(X_3, x)\right)$, is a non-degenerate random variable, the test for the test problem \eqref{eqn::testprob} is consistent, i.e.
		given a predefined level $\alpha \in [0, 1)$ and the corresponding critical value $c_{\frac{\alpha}{2}} $, it holds that 
		\begin{align}\label{eq:consistency}
			\Prob\left(\left|\frac{T^{(n,m)} - \frac{1}{2} \# I^{(n,m)} }{\sqrt{\Var_{H}(T^{(n,m)} )}} \right| > c_{\frac{\alpha}{2}} \right) \rightarrow 1, \ \text{as $n \rightarrow \infty$,}
		\end{align}
		where $\Var_{H}(T^{(n,m)} )$ denotes the variance under the null hypothesis as defined in \eqref{eqn::VarTest}. 
	\end{Theorem}
	
	\begin{bem}
		According to Theorem \ref{thm::Var}, replacing $\Var_{H}(T^{(n,m)})$ in \eqref{eq:consistency}  by $\frac{3}{16}\frac{n+m-6}{m-1} \# C_{i_1 = j_1} (2P_1-1)$ for any $P_1\in (1/2, 2/3)$ preserves consistency of a corresponding test.
	\end{bem}

	\section{Finite sample performance} \label{Sec:Monte Carlo}
	
	\subsection{Outline}
	
	\begin{figure}[h]
		\centering
		\includegraphics[width=\textwidth]{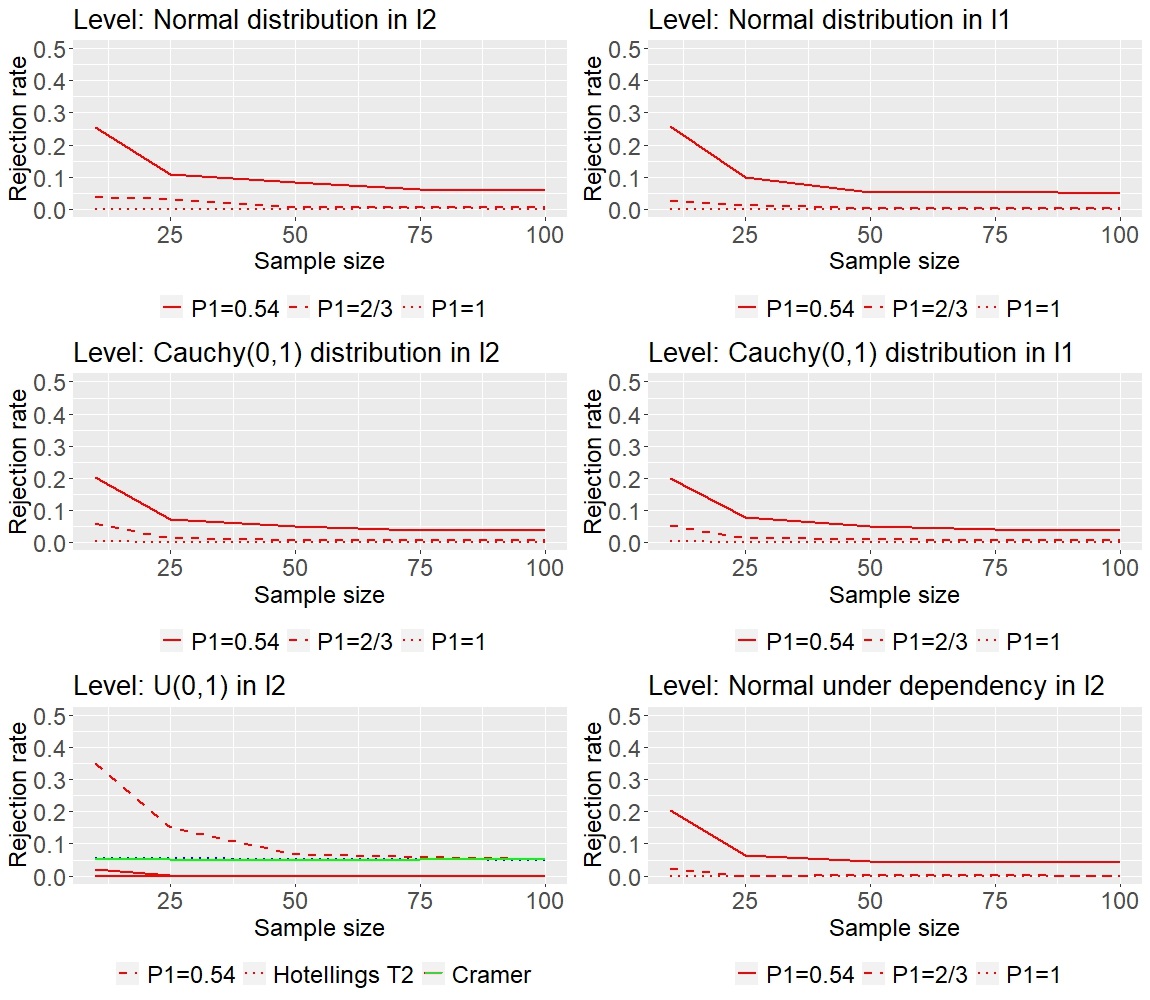}
		\caption{Empirical level of the tests $\Phi^d_1$ (red dotted), $\Phi^d_{2/3}$ (red dashed), $\Phi^d_{\rm oracle}$ (red full line). $d$ is either Euclidean metric ($l^2$) or Manhattan ($l^1$) on $\R^{10}$. The underlining distributions of the point clouds are standard normal ($\hat{P}_{1, \rm oracle} = 0.539$ and $\hat{P}_{1, \rm oracle} = 0.535$ for $l^2$ and $l^1$ oracle tests, respectively) in the upper left and right graphs, Cauchy with location $0$ and scale $1$ for  $l^2$ (middle, left) and $l^1$ (middle, right), uniform distribution on $\R^2$ over $l^2$ (bottom, left), with addition of Cram\'{e}r (green)  and Hotelling's $T^2$ (black) levels and normal distribution under dependency over $l^2$ (bottom right). } \label{fig::Null}
	\end{figure}
	\noindent
	In this section, we assess the finite sample performance of the hypothesis test established in Section \ref{Sec:Res}.  For this, we compute  its empirical size and power in different scenarios and under  varying  distributions of the data. The general form of an averaged Wilcoxon $\alpha$-level test is given by
	\begin{equation} \label{eqn::our_test}
		\Phi_{\hat{P}_1 }^d = \begin{cases}
			1, \textnormal{ for } \left|\tfrac{T^{(n,m)} - \frac{1}{2}\# I^{(n,m)}}{ \sqrt{\frac{3}{16} 
					\frac{n + m -6}{m-1}\# C_{i_1 = j_1}(2\hat{P}_1-1)}}\right| > c_{\frac{\alpha}{2}} \\
			0, \textnormal{ otherwise }
		\end{cases},
	\end{equation}
	where the statistic $T^{(n,m)}$ is defined by $T^{(n,m)} =  \sum_{\mathbf{i} \in I^{(n,m)}} \I_{\{d(X_{i_1}, X_{i_2}) \leq d(X_{i_3}, Y_{i_4})\}}$, $\# C_{i_1 = j_1} = \frac{1}{3}nm(m-1)(n-1)(n-2)(n-3)(n-4)$, $c_{\frac{\alpha}{2}}$ is the $(1- \frac{\alpha}{2})$-quantile of the standard normal distribution, and $I^{(n,m)} = \{ \mathbf{i}\in \N^4 | i_1, i_2, i_3 \in \left[n\right], i_4 \in \left[ n+m\right] \setminus \left[ n \right], \, i_1 < i_2, i_1 \neq i_3, i_2 \neq i_3 \}$. 
	The superscript $d$ of the test $\Phi_{\hat{P}_1 }^d$ indicates the metric or distance function, which enters the test via the test statistic $T^{(n,m)}$. The quantity $\hat{P}_1$ denotes an estimator for the probability $P_1=\Prob (\I_{\{d(X,X_2)\leq d(X_3, X_4)\}} = \I_{\{d(X,X_5)\leq d(X_6, X_7)\}})$ defined in Theorem \ref{thm::Var}, which is used to estimate the variance of the test statistic. 
	We investigate the performance of the averaged Wilcoxon test \eqref{eqn::our_test}
	(i) for the trivial  bound $\hat{P}_1= 1$, (ii) for the universal bound established in Theorem \ref{thm::Var}$, \hat{P}_1 = 2/3$, and (iii) for the appropriate oracle values, i.e. given known $P_1$, listed in Tables \ref{tab::Normal} and \ref{tab::Cauchy} in Section \ref{sec:numerical P1}. In case the oracle value is used, the distance $d$ also enters the variance estimation as the values for $P_1$ vary (slightly) from distance to distance. The three specific averaged Wilcoxon tests considered here are denoted by $\Phi^d_1$, $\Phi^d_{2/3}$ and $\Phi^d_{\rm oracle}$, where the latter is to be understood in the context of a given test problem, i.e. with $\hat P_1=\hat P_1(d,D)$ varying between distances and dimensions. \\
	\noindent
	Size, power and robustness of the three versions of $\Phi_{\hat P_1}^d$ with regard to the dimension of the datasets have been examined in Monte Carlo simulations, on normally, uniform and Cauchy distributed data, in location and scale problems and with respect to the  Euclidean metric, the $l^1$ (``Manhattan'') metric, and the Canberra distance function. Our simulations show that the oracle choice for $P_1$ underestimates the variance such that the nominal level is not maintained, while the upper bounds overestimate the variance and cause a loss of power. The performance of the three tests is compared to three other multivariate two-sample tests. The first is an independent version of the Wilcoxon rank sum test over distances, where the generated sample of distances has been randomly and independently sampled, with critical values based on the normal approximation. The second is Hotelling's $T^2$, a multivariate $t$-test generalization based on Hotelling's distribution (see \cite{anderson1958introduction}, pp. 171-177) and lastly, Cram\'{e}r's test with Monte-Carlo-bootstrapped critical values and Cram\'{e}r kernel, based on Euclidean interpoint distances; see \cite{baringhaus2004new} and \cite{franz2019package}. These three tests are denoted by $\Phi_{\rm ind}^d$, $\Phi_H$ and $\Phi_C$, respectively. \\
	\noindent
	As expected, the averaged Wilcoxon test clearly outperforms $\Phi_{\rm ind}^d$ in all scenarios. Further, in contrast to $\Phi_H$ and $\Phi_C$, the averaged Wilcoxon tests are suitable for applications on high-dimensional data and show  greater robustness to the deterioration of power with growing dimension. If not otherwise specified, the dimension of the data in the test problems is set to $10$, samples are always balanced (i.e. $n=m$), the components of the data vectors are independent and the number of Monte Carlo iterations is $1000$. In the graphs in Figures \ref{fig::Null}-\ref{fig::Powerl1}, the full red line always represents the oracle test, the dashed red line represents $\Phi_{2/3}$, the dotted red line represents $\Phi_1$ and the dashed blue line represents the independent Wilcoxon test.\\
	\begin{figure}[h]
		\centering
		\centering
		\includegraphics[width=\textwidth]{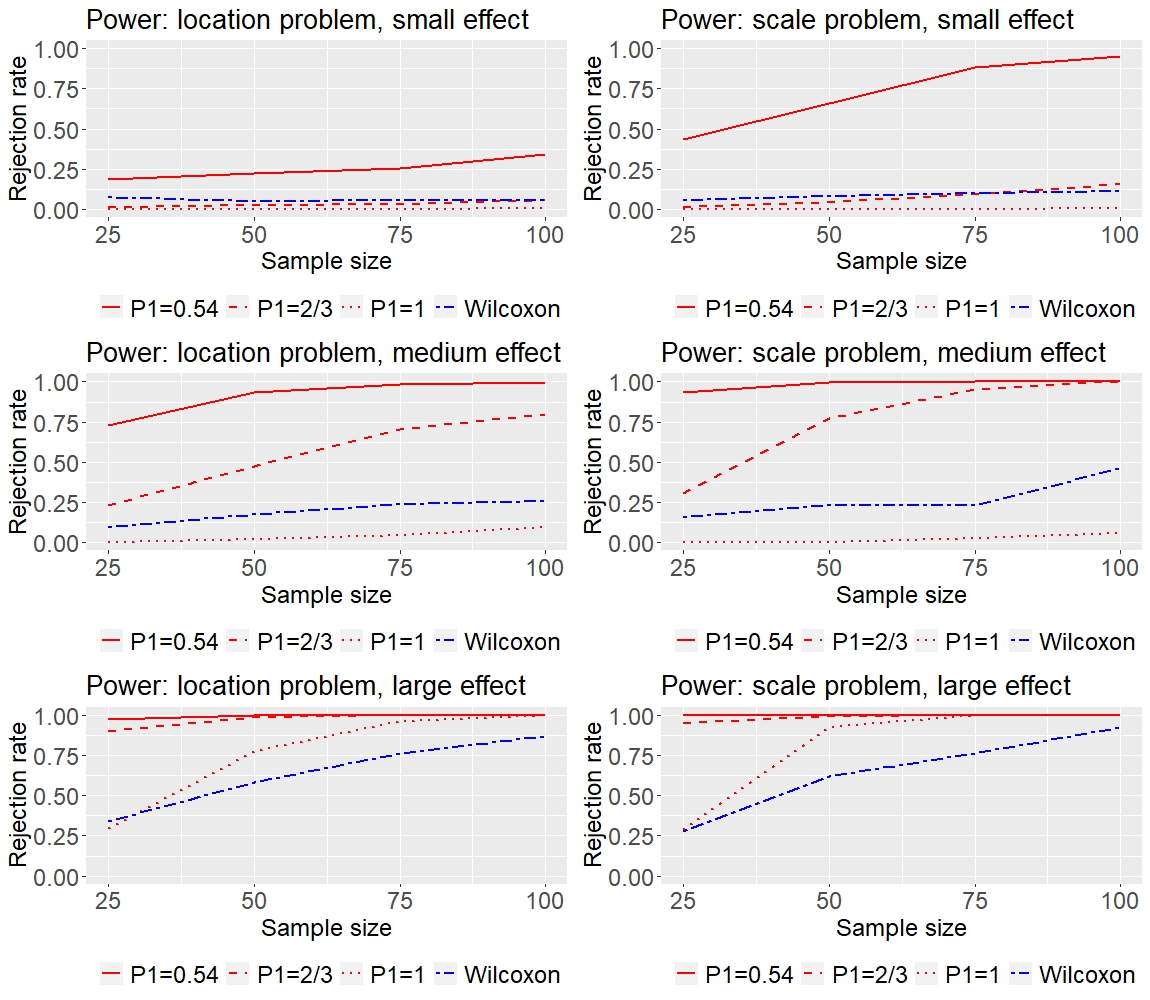}
		\caption{Power of $\Phi^d_1$ (red dotted), $\Phi^d_{2/3}$ (red dashed), $\Phi^d_{\rm oracle}$ (red full line) and independent Wilcoxon $\Phi^d_{\rm ind}$ (blue dotted line), where $d$ is the Euclidean metric on $\R^{10}$,  in location (left) and scale problems (right) for the normal distribution. Effect sizes: small ($\Delta \mu = 0.3$ and $ s=0.9$), medium ($\Delta \mu = 0.6$ and $ s=0.8$) and large ($\Delta \mu = 0.9$ and $ s=0.6$), where $\sigma_Y^2 = s \sigma_X^2$.} \label{fig::Powerl2}
	\end{figure}

	\subsection{Empirical level}
	\noindent
	In this section we investigate if the tests $\Phi_1^d, \Phi_{2/3}^d$, and $\Phi_{\text{oracle}}^d$ hold their nominal level, 
	which was set to $0.05$ for all simulations presented in  this section.\\
	
	\noindent
	\textit{Independent Components.} \label{sec:level_ind} To study the empirical level of the tests $\Phi_1^d, \Phi_{2/3}^d$, and $\Phi_{\text{oracle}}^d$ we simulate i.i.d.\ data of a fixed dimension $$X_1,\ldots,X_{n},Y_{n+1},\ldots,Y_{2n}\in\R^{10},$$
	where $X_i=(X_{i, 1}, \ldots, X_{i, 10})^{T}$ and  $Y_i=(Y_{i, 1}, \ldots, Y_{i, 10})^{T}$ and for two different scenarios:  $X_{1,j}, Y_{1,j}\sim\mathcal{N}(0,1)$, $j=1,\ldots,10$, for the first scenario,    $X_{1,j}, Y_{1,j}\sim\text{Cauchy}(0,1)$, $j=1,\ldots,10$, for the second scenario. In both cases, the components of the vectors are independent and $d$ is either the $l^1$ or the $l^2$ distance. The results are shown in Figure \ref{fig::Null} (upper and middle panels).
	The level of the test $\Phi_{\text{oracle}}^d$ is generally not upheld for sample sizes $\leq 50$. 
	The effect of underestimation of the variance using the oracle value for $ \hat P_1$ is apparent, especially for very small sample sizes:
	The effect is worse in the case of a normal distribution  than in the Cauchy example. We note the difference in level in the normal example depending on the choice between $l^1$ and $l^2$. We observe a  drop between sample sizes $10$ and $25$,   a phenomenon which can be explained by a significant underestimation of variance using the asymptotic expansion \eqref{eqn::VarTest} for small sample sizes.  \\

	\begin{figure}[h]
		\centering
		\begin{minipage}{\textwidth}
			\centering
			\includegraphics[width=1\linewidth]{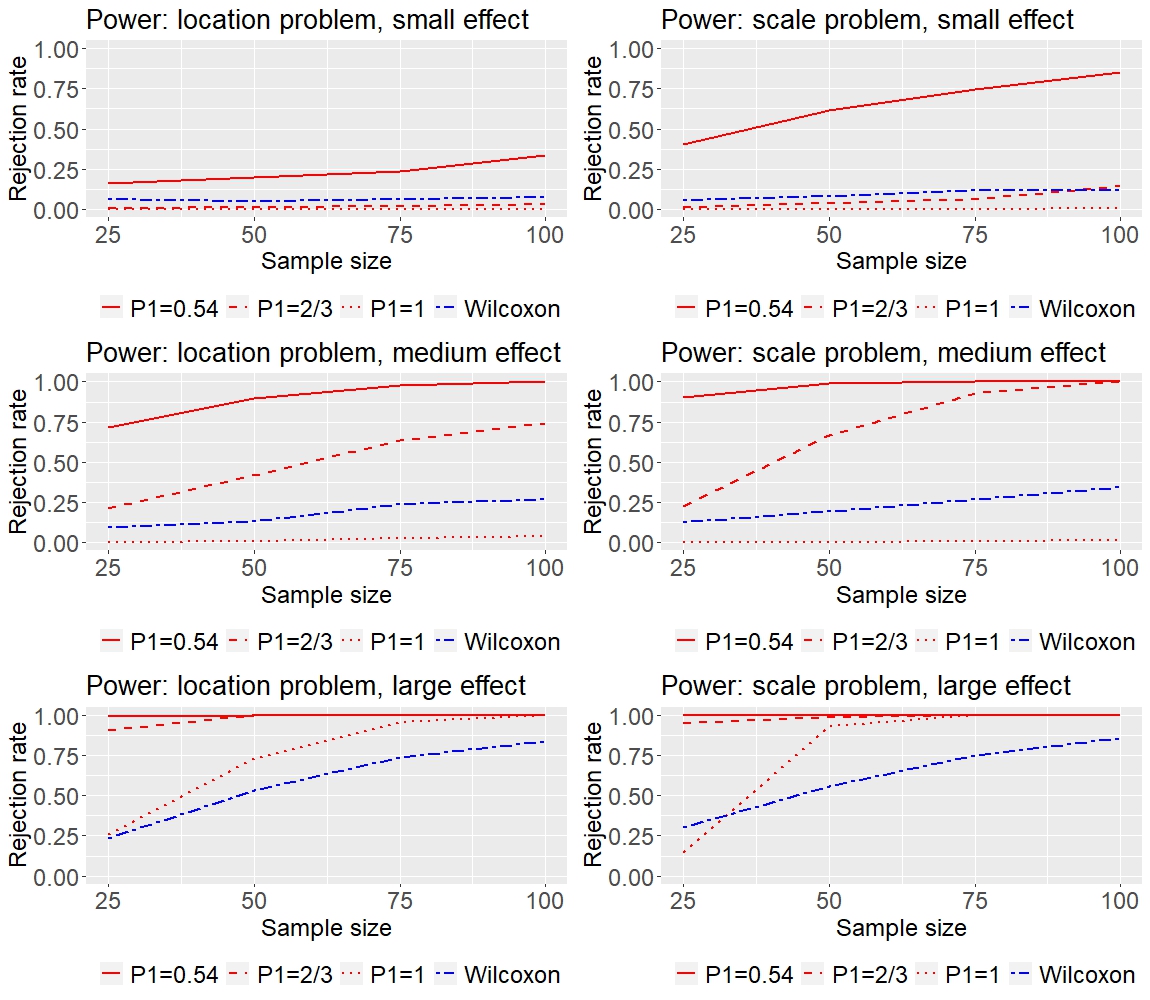}
			
		\end{minipage}%
		\caption{Power of the tests for the Manhattan ($l^1$) distance on $\R^{10}$. $\Phi^d_1$ (red dotted), $\Phi^d_{2/3}$ (red dashed), $\Phi^d_{\rm oracle}$ (red full line) and independent Wilcoxon $\Phi^d_{\rm ind}$ (blue dotted line),  in location (left) and scale problems (right) for the normal distribution. The data follows the same distributions as in the example given in Figure \ref{fig::Powerl2}. Effect sizes: small ($\Delta \mu = 0.3$ and $ s=0.9$), medium ($\Delta \mu = 0.6$ and $ s=0.8$) and large ($\Delta \mu = 0.9$ and $ s=0.6$), where $\sigma_Y^2 = s \sigma_X^2$.}  \label{fig::Powerl1}
	\end{figure}
	\noindent
	
	\noindent
	\textit{Dependent Components.}\label{bsp::dependence} We developed the averaged Wilcoxon test inspired by the application to genomic microarray datasets, where common (parametric) assumptions  are typically not met.
	For instance, in the data example  discussed in Section \ref{Sec:RealData}, it is unrealistic to assume that the components of the observed data vectors are independent. Our averaged Wilcoxon test does not pose any restrictive assumptions in this direction and is also applicable to dependent components. In this section, we therefore study the empirical level under dependence. \\
	\noindent
	We consider as example for $X_i, Y_i \in \R^{10}$ and $\Sigma = (\frac{1}{1+|i-j|})_{i,j =1..10}$
	\begin{equation*}
		X_i, Y_i \sim \mathcal{N}(0_{10}, \Sigma ),
	\end{equation*}
	where $0_{10}\in\R^{10}$ denotes the 10-dimensional zero vector.
	The oracle value in this example for dimension $10$ and Euclidean metric does not differ much from the independent case given in Table \ref{tab::Normal}. We estimate $\hat{P}_1 = 0.536$ for the oracle test.
	The underestimation of the variance for smaller sample sizes can be observed in this example as well. The estimated levels are given in Figure \ref{fig::Null}. The performance is slightly better than in the independent case.\\
	
	\noindent
	\textit{Uniform distribution.}\label{bsp::Uni} The last example we consider is given in $2$ dimensions for the uniform distribution, i.e.  
	\begin{equation*}
		X_i, Y_i \sim U[0,1]^2,
	\end{equation*}
	where $X_i, Y_i \in \R^{2}$.
	It is well known that parametric tests such as Hotelling's $T^2$ may not perform well if the underlying parametric assumptions are violated. Indeed, Hotelling's test does not uphold the nominal level in this example. The estimated levels of $\Phi_1^d, \Phi_{2/3}^d$,  $\Phi_{\text{oracle}}^d$, $\Phi_{H}$, and $\Phi_C$  are given in Figure \ref{fig::Null}.\\
	\begin{figure}[h]
		\centering
		\begin{minipage}{\textwidth}
			\centering
			\includegraphics[width=1\linewidth]{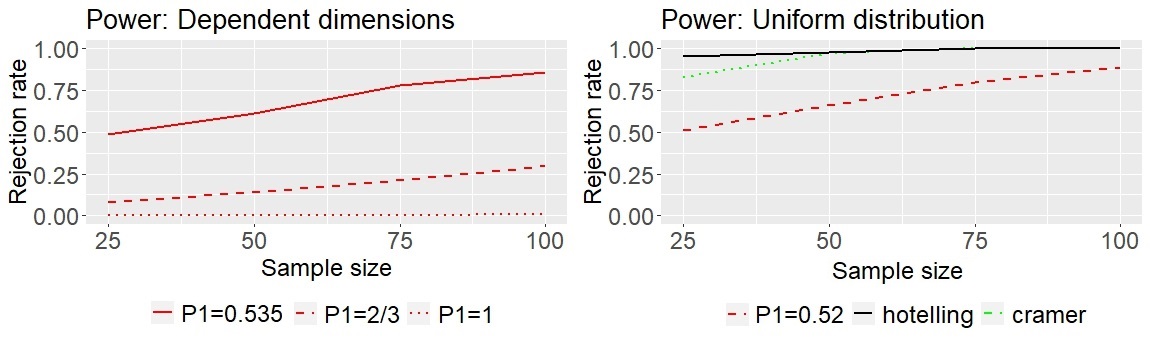}
			
		\end{minipage}%
		\caption{Power of the tests with respect to the Euclidean metric. Left: Normal dependent location problem in $10$ dimensions. Right: Uniform location problem in $2$ dimensions.  } \label{fig::PowerRest}
	\end{figure}
	\subsection{Power}
	\noindent
	In this section, we investigate the empirical power of all tests $\Phi_1^d, \Phi_{2/3}^d$,  $\Phi_{\text{oracle}}^d, \Phi^d_{\text{ind}}, \Phi_{H}$, and $\Phi_C$ in various scenarios. \\
	
	\noindent
	\textit{Power of $\Phi_1^d, \Phi_{2/3}^d, \Phi^d_{\text{ind}}$, and $\Phi_{\text{oracle}}^d$ in location and scale problems.}\label{subsec:power in oracle and scale} We first consider simple Gaussian location and scale problems. The location and scale differences are present in all the dimensions, i.e. for the location problem, we generate data according to
	\begin{equation}\label{eq:loc}
		X_i, Y_i\in\R^{10}, \quad X_i\sim \mathcal{N}(0_{10}, I_{10}),\quad Y_i \sim  \mathcal{N}(\Delta \mathcal{M}_i, I_{10}),
	\end{equation}
	where $\Delta \mathcal{M}_i = (\Delta \mu_i, ..., \Delta \mu_i) \in \R^{10}$, $\Delta \mu_i \in \{ 0.3, 0.6, 0.9\}$. Here and in the following, we denote by $0_{d}$ and $I_d$ the $d$-dimensional zero vector and identity matrix, respectively. For the scale problem, we consider 
	\begin{equation*}
		X_i, Y_i\in\R^{10}, \quad X_i\sim \mathcal{N}(0_{10}, I_{10}),\quad Y_j \sim  \mathcal{N}(0_{10}, sI_{10}),
	\end{equation*}
	where $s \in \{ 0.9, 0.8, 0.6 \}$.  All tests have been performed for Manhattan ($l^1$) and Euclidean ($l^2$) distances, where the corresponding oracle values from Table \ref{tab::Normal} have been used. The results are shown in Figures \ref{fig::Powerl2}-\ref{fig::Powerl1}. As expected, out of the three versions of the averaged Wilcoxon test, $\Phi_{\text{oracle}}^d$ uniformly has the greatest power, while $\Phi_{1}^d$ has the lowest.   $\Phi_{2/3}^d$ possesses non-trivial power in medium and large effect examples, while $\Phi_{1}^d$ possesses non-trivial power in problems with large effect size only. We observe a slight loss of power of tests for the choice of metric $l^1$ over $l^2$.\\
	
	\begin{figure}[h]
		\centering
		\begin{minipage}{\textwidth}
			\centering
			\includegraphics[width=1\linewidth]{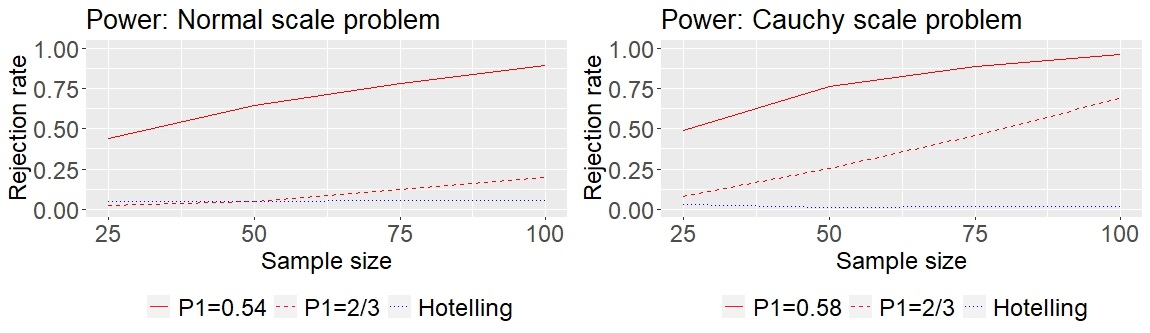}
			
		\end{minipage}%
		\caption{Comparison of averaged Wilcoxon tests and Hotelling's $T^2$. Left: Normal scale problem, small effect size ($s=0.9$, where $\sigma_Y^2 = s \sigma_X^2$). Right: Cauchy scale problem, large effect (scale in the alternative is $s = 0.5$). Both examples are performed over Euclidean metric in $\R^{10}$. } \label{fig::PowerHotelling}
	\end{figure}
	
	\noindent
	\textit{Deterioration of power in high dimensions.} In this section, we consider problems in which a difference in location and scale is only present in one component for growing dimension $D\in\{5, 10, 50, 200\}$. The sample sizes are between $25$ and $100$. To be precise, we consider 
	\begin{equation*}
		X_i, Y_i\in\R^{D}, \quad X_i\sim \mathcal{N}(0_{D}, I_{D}),\quad Y_{50+i} \sim  \mathcal{N}((0.6,0_{D-1}), I_{D}), \quad i=1,\ldots,50,
	\end{equation*}
	and 
	\begin{equation*}
		X_i, Y_i\in\R^{D}, \quad X_i\sim \mathcal{N}(0_{D}, I_{D}),\quad Y_{50+i} \sim  \mathcal{N}(0_{D}, \Sigma_{0.5}^D), , \quad i=1,\ldots,50,
	\end{equation*}
	where the covariance matrix $\Sigma_{0.5}^D$ is given by
	\begin{equation}\label{eqn::Y2}
		\Sigma_{0.5}^D= \begin{pmatrix}
			0.5 & 0   & 0   & \cdots & 0   \\
			0   & 1   & 0   & \cdots & 0   \\
			0   & 0   & 1   & \cdots & 0   \\
			\vdots & \vdots & \vdots & \ddots & \vdots \\
			0   & 0   & 0   & \cdots & 1
		\end{pmatrix},
	\end{equation}
	i.e. the effect is limited to one dimension. In this and related problems, a deterioration of power occurs naturally as the dimensions grow, which is also confirmed by our simulation results. The tests are performed over the Euclidean metric since it has shown to outperform the Manhattan metric in Section \ref{subsec:power in oracle and scale}. Figure \ref{fig::det}  illustrates the loss of power due to increasing dimensionality for the two scenarios.\\
	\begin{figure}[h]
		\centering
		\begin{minipage}{\textwidth}
			\centering
			\includegraphics[width=1\linewidth]{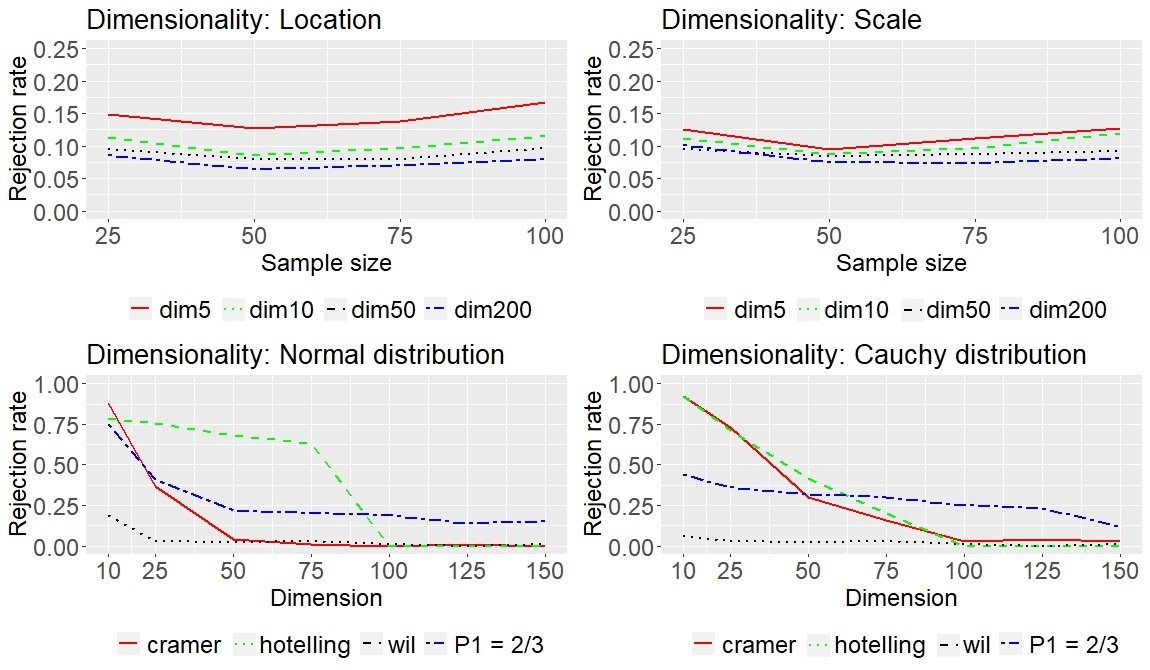}
		\end{minipage}%
		\caption{Deterioration of power of the averaged Wilcoxon test with the dimension. Top: $\Phi_{2/3}^d$ in location (left, $\Delta \mu = 0.6$) and scale (right, $\sigma_Y = 0.6 \sigma_X$) examples with growing dimension $D$, where $D= 5$ (red), $10$ (green), $50$ (black) and $200$ (blue line). Bottom: Comparison of power degradation for Cram\'{e}r's, Hotelling's and independent Wilcoxon's tests. } \label{fig::det}
	\end{figure}
	\noindent
	\begin{figure}[h] 
		\centering
		\begin{minipage}{\textwidth}
			\centering
			\includegraphics[width=1\linewidth]{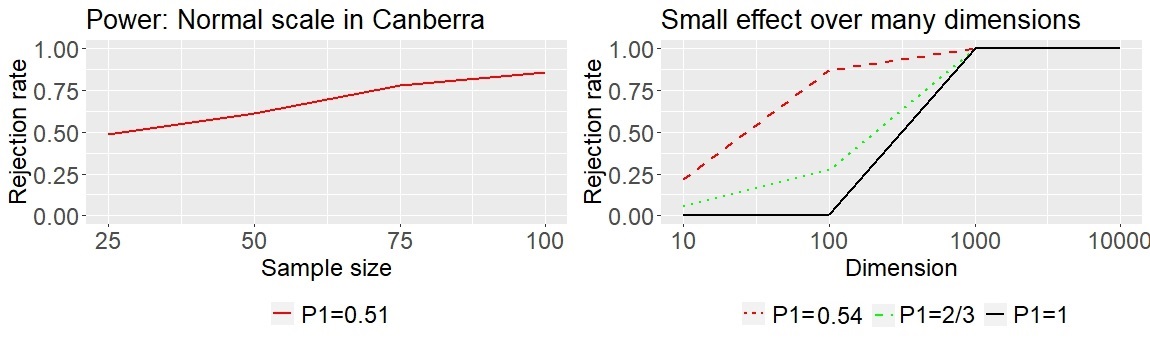}
			
		\end{minipage}%
		\caption{Left: Averaged Wilcoxon test over Canberra distance for scale problem, normal distributions with $\sigma_2 = 0.5\sigma_1$ in $10$ dimensions. Right: Empirical power for small effect size over all dimensions growing with the dimension, for the three values of $P_1$, for the normal distribution, log-scaled.}  \label{fig::Canberra}
	\end{figure}
	\noindent
	The deterioration of power due to the consideration of higher dimensions has been further examined and compared to other tests. The results of this set of simulations are presented in Figure \ref{fig::det}, where greater stability of power over dimensions has been observed for the averaged Wilcoxon test. The test $\Phi_{2/3}$ has been compared with Hotelling's $T^2$, Cram\'{e}r's and independent Wilcoxon tests, where the curves are denoted in blue, green, red and black, respectively. The tests have been compared in the location problem \eqref{eq:loc} and, additionally in a  Cauchy location problem with
	\begin{equation}
		X_i\sim \textnormal{Cauchy}(0, 1)^{D}, \quad Y_{50+i} \sim \textnormal{Cauchy}(1, 1) \times \textnormal{Cauchy}(0, 1)^{D-1}\quad i=1,\ldots,50.
	\end{equation}
	The sample size $n=m=50$ is fixed and the dimensions between $10$ and $150$. The independent Wilcoxon test displays similar stability over dimensions to averaged Wilcoxon, albeit with significantly lower power. Although Cram\'{e}r's and Hotelling's tests have higher power in problems of smaller dimension, these tests lose power rapidly and are less robust to dimensionality than the averaged Wilcoxon test.\\ 
	
	\noindent
	\textit{Dependence.} We consider further the example given in \ref{bsp::dependence} for the alternative
	\begin{equation*}
		Y_i \sim \mathcal{N}(\Delta \mathcal{M}, \Sigma ),
	\end{equation*}
	where $\Delta \mathcal{M} = (0.6,..., 0.6) \in \R^{10}$, i.e. the normal location problem under dependence. The chosen oracle value is $\hat{P}_1 = 0.536$ and the tests are performed with respect to the  Euclidean distance. The results are given in Figure \ref{fig::PowerRest}. We observe uniformly slight loss of power compared to the independent case. \\
	
	
	\noindent
	\textit{Uniform example.} For the uniform distribution considered in Section \ref{bsp::Uni}, we consider the alternative  
	\begin{equation*}
		Y_i \sim U[0.2,1.2]^2
	\end{equation*}
	and perform $\Phi_{\text{oracle}}^d$ with respect to the Euclidean metric. The power of the test is compared to Hotelling's and Cram\'{e}r's tests. The results are given in Figure \ref{fig::PowerRest}.\\
	
	
	\noindent
	\textit{Power of the test for a non-standard distance.} A potentially interesting distance function for the application to genomic microarray data is the Canberra distance, which proved to be well suited for distinguishing between groups of individuals based on their micro RNA profiles in hierarchical clustering algorithms (cf. \cite{proksch2023personalised}). Defined by $d(X,Y) := \sum_{i=1}^D \frac{|X_i - Y_i|}{|X_i|+|Y_i|}$, this quasi-normed distance equalizes the contribution of individual dimensions to the overall distance between points. Unfortunately, the Canberra distance is not a metric and in particular does not fulfill  condition (D2) of  Maa's theorem. We  therefore do not have a guarantee for the well-posedness of a test problem defined with respect to  Canberra distances. The simulations also indicate that the location test over Canberra distance has almost no power in the examples considered here. However, it seems that it is possible to distinguish  scale differences well, as illustrated in Figure \ref{fig::Canberra}. In this example, we test the scale difference of normal distributions with $\sigma_2 = 0.5\sigma_1$.  Choosing $\hat{P}_1=1$ and $\hat{P}_1=2/3$ for the variance approximation grossly overestimates the variance due to the low $P_1$ oracle value (cf. Table \ref{tab::Normal}). As a result,  the corresponding test has low power. Nonetheless, the oracle value $\hat{P}_1 = 0.5063$ does have nontrivial power, further illustrating the necessity for a better variance approximation. \\
	\noindent
	Lastly, Figure \ref{fig::Canberra} visualises the examination of power of the three versions of the test for a small effect size over a large number of dimensions, i.e.
	\begin{equation*}
		Y_i \sim  \mathcal{N}(\Delta \mathcal{M}, I_{D}),
	\end{equation*}
	where $\Delta \mathcal{M} = (0.3, ..., 0.3) \in \R^D$ with the dimension $d$  in the range from $10$ to $10000$. The graph is given in log scale. The convergence of power to $1$ as the dimension goes to infinity can be observed. However, for the example considered here, the $l^1$ and $l^2$ distance outperform the test based on pair-wise Canberra distances.\\
	\FloatBarrier

	\section{Real Data Example}\label{Sec:RealData}
	\noindent
	Simulated examples indicate the stability of the power of averaged Wilcoxon tests with increasing dimension, making the test a viable choice for high-dimension low-sample-size problems (HDLSS). An example of such settings are genetic studies in which the goal is to detect the genes most differentially expressed between  groups from diverse populations. Contrasting groups can give an indication of underlying mechanisms that cause the group differentiation, e.g. indicate gene-disease associations or genetic characteristics of groups. Genomic microarray datasets usually possess large number of features (dimensions, genes) while at the same time having low sample sizes, especially relative to their dimension. On the one hand, performing such measurements is still very costly and complex. On the other hand, one fundamental limitation causing this predicament is low population size, specifically, low number of people with a particular disease, property or a set of properties, relative to the number of observed variables, which emphasizes the relevance of methods tailored to HDLSS settings. \\
	\noindent
	The following practical example is based on the dataset corresponding to the study \cite{Lung} examining micro RNA (miRNA) levels in blood plasma of patients with lung adenocarcinoma. The dataset consists of raw counts of originally $1509$ miRNAs (features in this example) in $57$ samples across $6$ groups. Groups are labeled "granuloma", "cancer", "exosomes granuloma", "exosomes cancer", "without exosomes granuloma" and "without exosomes cancer". The study aims at differentiating between the patients with lung cancer and misdiagnosed patients with benign conditions.  The samples of the cancer groups and the rest of the samples are, for the purposes of this example, pooled resulting in sample sizes $n = 30$ (pooled cancer) and $m = 27$ (pooled healthy). The measurements are represented in \textit{counts per million} (CPM). CPM is a normalization method for RNA sequencing data that represents the number of reads of particular RNA scaled by the total number of reads in the sample, then multiplied by one million to facilitate comparisons between samples. Ad hoc dimensionality reduction has already been performed by including only the miRNAs which have counts per million larger than $1$ in more than $1$ sample, excluding the miRNAs that have low concentrations and are rarely measured (\cite{anders2013count}, p24). Only the miRNA with CPM larger than $1$ have been included in the dataset. \\

	\begin{table}[h!]
		\centering
		\begin{tabular}{||c c c c c||} 
			\hline
			miRNA       &    log-fold change  & log-CPM   &    p-value   &      BH correction \\ [0.5ex] 
			\hline\hline
			
			hsa-miR-3529-3p &  2.272515 & 9.817300 & 4.820727 $\cdot 10^{-8}$ & 7.274476$\cdot 10^{-5}$ \\
			hsa-miR-7-5p   &   2.131011 &9.013191 &6.358140 $\cdot 10^{-7}$ &4.797217$\cdot 10^{-4}$ \\
			hsa-miR-576-3p  &  1.590984 &9.030031 &1.123073$\cdot 10^{-5}$ &5.649059$\cdot 10^{-3}$ \\
			hsa-miR-1827     & 1.673530 &4.110844 &2.038026$\cdot 10^{-5}$ &6.829125$\cdot 10^{-3}$ \\
			hsa-miR-4306      &1.948530 &8.998243 &2.262798$\cdot 10^{-5}$ &6.829125$\cdot 10^{-3}$ \\
			hsa-miR-9-5p      &2.400299 &7.625194 &5.718436$\cdot 10^{-5}$ &1.351211$\cdot 10^{-2}$ \\
			hsa-miR-501-5p    &3.018604 &4.190806 &7.337386$\cdot 10^{-5}$ &1.351211$\cdot 10^{-2}$\\
			hsa-miR-629-3p    &3.218628 &3.371337 &7.622035$\cdot 10^{-5}$ &1.351211$\cdot 10^{-2}$\\
			hsa-miR-1180-3p   &1.384724 &9.953922 &9.180467$\cdot 10^{-5}$ &1.351211$\cdot 10^{-2}$\\
			hsa-miR-92b-5p    &1.718041 &6.529031 &9.441425$\cdot 10^{-5}$ &1.351211$\cdot 10^{-2}$\\
			hsa-miR-218-5p    &2.803250 &7.846570 &9.849784$\cdot 10^{-5}$ &1.351211$\cdot 10^{-2}$\\
			hsa-miR-3158-3p   &1.702470 &5.964335 &1.371379$\cdot 10^{-4}$ &1.724509$\cdot 10^{-2}$\\
			hsa-miR-181a-2-3p &1.221580 &8.373332 &3.234169$\cdot 10^{-4}$ &3.754124$\cdot 10^{-2}$\\
			hsa-miR-500a-5p   &5.036958 &3.772247 &4.843275$\cdot 10^{-4}$ &5.183284$\cdot 10^{-2}$\\
			hsa-miR-877-5p    &1.638779 &5.765388 &5.152369$\cdot 10^{-4}$ &5.183284$\cdot 10^{-2}$\\ [1ex]
			\hline
		\end{tabular}
		\caption{Differentially expressed micro RNAs with their log fold changes, log-counts per million averages, p-values of the exact test and Benjamini-Hochberg corrected p-values.}
		\label{table::miRNAs}
	\end{table}
	\noindent
	It is important to note that the test results over pooled groups are not necessarily biologically meaningful. However, the example provided is self-consistent in the sense that the results of  standard differential expression tools as used by bioinformaticians (e.g. exact tests over normalized counts with multiplicity correction, see \cite{robinson2010edger})
	have been confirmed by using averaged Wilcoxon-based tests. The example furthermore illustrates well the necessity of testing at least two out of three equalities of distance distributions from Maa's theorem.   \\
	\noindent
	Differential expression analysis has been performed on the pooled dataset using the \texttt{edgeR} package in \texttt{R}. The in the package implemented procedure uses a locally exact test for differences in mean between two groups, which assumes a negative binomial distributions of miRNA counts and adjusts $p$-values according to the Benjamini-Hochberg correction. The first $15$ most differentially expressed miRNAs are given in Table \ref{table::miRNAs}, reported with log-fold changes and log-CPM for completeness. The procedure makes $13$ discoveries based on adjusted $p$-values below $0.05$.\\

	\begin{figure}[h] \label{fig::hist}
		\centering
		\begin{minipage}{\textwidth}
			\centering
			\includegraphics[width=1\linewidth]{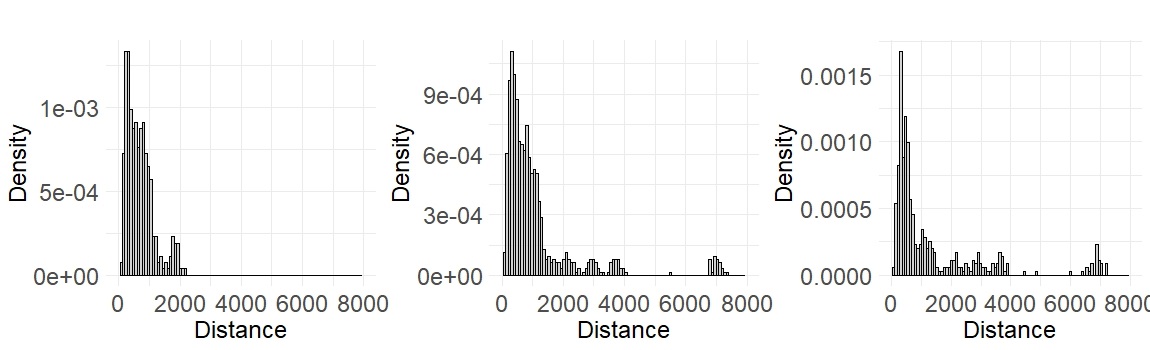}
			
		\end{minipage}%

		\caption{Normalized histograms of $d_{XX}$ (left), $d_{XY}$ (middle) and $d_{YY}$ (right)-example of distance distributions over the panel containing first $6$ differentialy expressed miRNAs (1..6 of Table \ref{table::averaged test}).} \label{fig::Dim}
	\end{figure}

	\begin{table}[h!]
		\centering
		\begin{tabular}{||c c c c c||} 
			\hline
			Dimensions & $W_1$ & p & $W_2$ & p \\ [0.5ex] 
			\hline\hline
			1.2 & 3.13356... &0.00086 & -2.06327... &0.01954 \\ 
			1.3 &  1.49720... & 0.06717 & -2.62111... &0.00438 \\
			1.4 & 1.49987... & 0.06682 & -2.62000... &0.00439 \\
			1.5 & 2.18150... & 0.01457 & -3.14900... & 0.00081\\
			1.6 & 2.15037... & 0.01576 & -3.15471... & 0.00080 \\
			1.7 & 2.13957... & 0.01619 &  -3.15738... & 0.00079 \\
			1.8 & 2.14050... & 0.01615 & -3.15757... & 0.00079\\
			1.9 & 1.26487... & 0.10295 & -2.94173... & 0.00163\\
			1.10 & 1.25779... & 0.10423 &-2.94330 ... & 0.00162 \\ [1ex]
			\hline
		\end{tabular}
		\caption{Values of the test statistics and corresponding p-values for the test problems $\mathcal{H}_{0,1}: d_{XX} \sim d_{XY}$ (columns $2$ and $3$) and $\mathcal{H}_{0,2}: d_{YY} \sim d_{XY}$ (columns $4$ and $5$) over the Euclidean distance and relevant dimensions implied by the sequence of tests given in Table \ref{table::miRNAs}}
		\label{table::averaged test}
	\end{table}
	\noindent
	These results are then compared with the test $\Phi_{2/3}$ over the Euclidean distance for the test problems $\mathcal{H}_{0,1}: d_{XX} \overset{\mathcal{D}}{=} d_{XY}$ and $\mathcal{H}_{0,2}: d_{YY} \overset{\mathcal{D}}{=} d_{XY}$, where $X$ represents the healthy group and $Y$ the sick group. Note that the results of the tests are similar over other $l^p$-induced metrics. The dimensions are added stepwise according to the list obtained by differential expression analysis in Table \ref{table::miRNAs}. The test of the first hypothesis detects the difference in distributions of $d_{XX}$, the distance between the points in the healthy group and $d_{XY}$, the distance between the points in sick and healthy groups,  while the second compares the distances between the sick $d_{YY}$ with $d_{XY}$. The test statistics and corresponding p-values are given in Table \ref{table::averaged test}. The second test confirms the first $10$ discoveries, while the first test does not reject for the combinations of first $3$, $4$, $9$ and $10$ features.\\

	\section{Outlook} \label{Sec:Conclusion}
	
	With this article, we introduced a Wilcoxon-inspired rank-based test for two-sample homogeneity test problems over interpoint distances, which provides a straightforward way for dimensionality reduction. The advantages of the test, which inherits some of its properties from the Wilcoxon test, include weak regularity conditions on the underlying distributions enabling comparisons for location, scale and shape problems, readily available critical values, and robustness against high-dimensionality. A procedure revolving around this test would include the choice of a metric or  dissimilarity function between the points tailored to the structure of the data. We  examined the power of the test in different situations and illustrated its use on  microarray data, all of which underlines the  potential  of the proposed test. While,  for  simplicity of the argument, we  concentrated on  situations for which ties in distance comparisons do not occur, we believe that the results of this paper can be generalized to cases where ties are possible (discrete distributions and corresponding, suitably chosen metrics). A discrete version of  Maa's theorem provides a basis for such an approach. While the dependence structure in the case allowing for ties  changes, the orders of the dependency neighbourhoods do not, as well as the orders of other relevant quantities in Theorem \ref{thm::CLT} needed to ensure a Gaussian limit. Naturally and notably, allowing for ties would necessitate a refined  variance estimation for the implementation of the test.  \\
	We further illustrated the loss of power in low sample size situations due to the choice of an estimator for the probability $P_1$  in the variance approximation. The variance is, on one hand,   underestimated by the use of oracle  values for $P_1$ and, on the other hand, overestimated by the general upper bound provided by Theorem \ref{thm::P2}. The power of the test could further be increased in low sample size situations by improving the variance estimate, foremost by tightening the bound on $P_1$, but also by including higher-order terms in the asymptotic expansion of the variance. The probability $P_1$ is an object interesting on its own merit. While the probability $P_2$ is a universal quantity over all underlying distributions and metric choices, $P_1$ depends on dimension, distribution and metric, as demonstrated by numerical simulations and indicated by two specific examples considered in this article.
	We conjecture that a tighter upper bound, dependent on dimension and metric, can be established that clarifies the interplay of these parameters. However, this seems to be highly non-trivial and will be considered in future research.
	On the other side, the universality of the lower bound on $P_1$ is not clear. By means of Theorem \ref{thm::CLT} we cannot guarantee the convergence of the test statistics for  distributions for which $P_1=\frac{1}{2}$ (given that these exist). The characterization of $P_1$ is an interesting problem that we intend to study further. \\
	Compared to some competitors computing the test statistic is computationally expensive.
	This arises from the necessity of making all  possible comparisons between induced distances, which, given balanced samples, results in a computational cost of order $n^4$.
	For comparison, note that the energy distance based test and Cram\'{e}r test are of the order $n^2\log(n)$. Nonetheless, it is important to note that no further computations are necessary in order to apply the test since, when using the universal bound on the variance, the critical values are readily available and do  not require expensive resampling approaches.
	Despite the advocated use case for the test revolving around the low sample size setting, this inherent cost plays a significant role if the test is to be integrated into a multiple testing framework, as illustrated by the practical differential expression example.\\
	Another interesting question warranting further investigation is the suitability of the distribution of non-standard and pseudo differences, such as the Canberra distance, which is of special interest in applications for the purpose of homogeneity testing. Namely, the theorem of Maa et al. does not guarantee the equivalence of testing for equal distributions over the Canberra distance function, since it does not fulfill the scalability and translational invariance property (D2). The arguments for the lower bound on $P_1$ given in the proof of Theorem \ref{thm::P2}  do not apply to the  Canberra distance function, either. However, it has been indicated that the test based on the  Canberra distance could be sensitive to scale problems. Under which additional conditions on the underlying distributions of $X$ and $Y$  the theorem of Maa et al. remains valid should be considered as a question for future research.\\

	\section*{Acknowledgements} 
	Annika Betken
	gratefully acknowledges financial support from the Dutch Research Council (NWO) through VENI grant 212.164. 
	
	\section{Appendix}\label{Sec:Appendix}
	
	\subsection{Proof Theorem \ref{thm::Tconv}}\label{sec:Tconv}
	$(\# I^{(n,m)})^{-1} T^{(n,m)}$ corresponds up to a multiplicative constant to the second order $U$-statistics with kernel $\I_{\{d(X_{i_1}, X_{i_2} ) \leq d(X_{i_3}, Y_{j_1}) \}}$ of the degree $(3,1)$. By the strong law of large numbers for multisample $U$-statistics (e.g. Theorem 3.2.1, p.97 in \cite{korolyuk2013theory}), we have the convergence.
	
	\subsection{Proof Theorem  \ref{thm::Var}}
	Since, by assumption, ties in the distance comparisons considered here almost surely do not occur,
	the second sum in (\ref{eqn::TT}) is constant almost surely. Therefore, the test statistic $T^{(n,m)}$ and its variance correspond to (\ref{eqn::Tfinal}) and (\ref{eqn::VarT(n,m)}), respectively, and we find
	\begin{equation}\label{eqn::VarT1}
		\Var (T^{(n,m)}) = \frac{1}{4} \# I^{(n,m)} + \frac{1}{4} \sum_{(\mathbf{i}, \mathbf{j}) \in E_{(n,m)}}  \left( 2\Prob (\I_{\mathbf{i}} = \I_{\mathbf{j}}) - 1 \right),
	\end{equation}
	where $ I^{(n,m)} = \{ \mathbf{i}\in \N^4 | i_1, i_2, i_3 \in \left[n\right], i_4 \in \left[ n+m\right] \setminus \left[ n \right], \, i_1 < i_2, i_1 \neq i_3, i_2 \neq i_3 \}$ and
	$E_{(n,m)}=\{\mathbf{i}, \mathbf{j} \in I^{(n,m)}\,|\,\mathbf{i}\cap\mathbf{j}\neq\emptyset , \mathbf{i}\neq\mathbf{j}\}$. The asymptotic order of the first summand, which corresponds to the diagonal elements of the variance, is hence given by $\# I^{(n,m)} \sim n^4$, i.e. $\# I^{(n,m)} / n^4 \rightarrow C < \infty$ for some constant $C>0$. Note that $n \sim m$ by the assumption. In order to determine the (order of the) asymptotically leading term, we will now decompose the index set $E_{(n,m)}$ according to the amount of shared indices: $\# \mathbf{i}\cap\mathbf{j}$. The most frequent occurrence is the case $ \# \mathbf{i}\cap\mathbf{j}=1$, i.e., most pairs $(\mathbf{i}, \mathbf{j}) \in E_{(n,m)}$ share exactly one index. We define the set $E_{(n,m),1}\subset E_{(n,m)}$ such that it contains exactly these indices and set $E_{(n,m),{\rm rest}}:=E_{(n,m)}\backslash E_{(n,m),1}$. We find that the number of the pairs $(\mathbf{i}, \mathbf{j})$ sharing two or more indices is of the order $n^6$: we have order of $\binom{n}{2}$ choices for the two fixed indices and order $\binom{n}{4}$ choices for the rest of the indices. There are less pairs sharing three indices and therefore $\# E_{(n,m),{\rm rest}}\sim n^6$. $E_{(n,m), 1}$ and $E_{(n,m)}$ are hence of the same asymptotic order, namely of order  $n^7$: we have $n$ choices for fixing one index and order of $\binom{n}{6}$ choices for the rest of the indices. The asymptotically dominating part of the sum in (\ref{eqn::VarT1}) is therefore that over pairs  $(\mathbf{i}, \mathbf{j})$ which share exactly one index. \\
	\noindent
	A pair $(\mathbf{i}, \mathbf{j})$ of multiindices sharing exactly one index corresponds to a probability $\Prob (\I_{\mathbf{i}} = \I_{\mathbf{j}})$, e.g., for the case $i_1 = j_1$ it is $\Prob (\I_{\mathbf{i}} = \I_{\mathbf{j}})= \Prob\big(\I_{ \{d(X_1,X_2)\leq d(X_3,X_4)\}}= \I_{\{d(X_1,X_5)\leq d(X_6,X_7)\}}\big) =: P_1 $. Since under the null hypothesis with no ties
	\begin{align*}
		P_1&=\Prob(\I_{ \{d(X_1,X_2)\leq d(X_3,X_4)\}}= \I_{\{d(X_1,X_5)\leq d(X_6,X_7)\}} )\\&= 1 - \Prob(\I_{ \{d(X_1,X_2)\leq d(X_3,X_4)\}}= \I_{\{d(X_5,X_6)\leq d(X_1,X_7)\}} ),
	\end{align*}
	there are only two possible values of the probability $\Prob (\I_{\mathbf{i}} = \I_{\mathbf{j}})$ if the multiindices share exactly one index (namely $P_1$ and $1-P_1$). Therefore, in the following we will rewrite the second term of the variance such that it only contains the terms $2P_1$ and $2P_1-1.$ The final task in this proof is to count the occurrences of each. 
	Notice that we have  $P_1 > \frac{1}{2}$ by assumption, which guarantees that the summands $2P_1 - 1$ in the second sum in (\ref{eqn::VarT(n,m)}) are positive, while $2(1-P_1) -1 = -(2P_1 - 1)$ are negative. In consequence there is a number of cancelling pairs in the sum.  All the possible configurations of such $(\mathbf{i}, \mathbf{j}) \in  I^{(n,m)} \times  I^{(n,m)}$ such that $\#(\mathbf{i}\cap\mathbf{j})=1$ with their probabilities expressed in terms of  $P_1$ are given in Table \ref{tbl::ij}. Note that  Table \ref{tbl::ij} is symmetric both in the sense of the corresponding probabilities as well as regarding the number of instances in each case. \\
	
	\begin{table}[h]

		\begin{center}
			\begin{tabular}{|c|c|c|c|c|}
				\hline
				$ \mathbf{j} \setminus \mathbf{i}$ & $i_1$ & $i_2$ & $i_3$ & $i_4$ \\
				\hline
				$j_1$ &  \cellcolor{green!25}$P_1$ & \cellcolor{green!25}$P_1$ & \cellcolor{red!25} $1-P_1$ & \cellcolor{gray!25} $\emptyset$ \\
				\hline
				$j_2$& \cellcolor{green!25}$P_1$ & \cellcolor{green!25}$P_1$ & \cellcolor{red!25} $1-P_1$ &  \cellcolor{gray!25} $\emptyset$\\
				\hline
				$j_3$& \cellcolor{red!25} $1-P_1$ & \cellcolor{red!25} $1-P_1$  & \cellcolor{green!25}$P_1$ &  \cellcolor{gray!25} $\emptyset$ \\
				\hline
				$j_4$& \cellcolor{gray!25} $\emptyset$ &  \cellcolor{gray!25} $\emptyset$ &  \cellcolor{gray!25} $\emptyset$ & \cellcolor{green!25}$P_1$ \\
				\hline
			\end{tabular}
			
		\end{center}
		\caption{Possible configurations of $\mathbf{i}$ and $\mathbf{j}$ sharing exactly one index and their corresponding probabilities $\Prob(\I_{\mathbf{i}} = \I_{\mathbf{j}})$. Intersection of column $i_p$ with row $j_q$ represents the case $i_p = j_q$. Green fields represent positive summands and the red ones negative.}\label{tbl::ij}
	\end{table}
	\noindent
	\paragraph{Counting the elements in the relevant index sets}~\\
	In the following, we count the number of all pairs sharing exactly one index case by case. 
	
	\paragraph{Case 1: $\mathbf{i_1 = j_1}$.}~\\
	The first case is $i_1 = j_1$, i.e. we will determine the cardinality $\# C_{i_1 = j_1},$ where $$C_{i_1 = j_1} := \{ (\mathbf{i}, \mathbf{j})\in I^{(n,m)}\times I^{(n,m)} | i_1 = j_1 \textnormal{ and } i_p \neq j_q \textnormal{ for } p,q =2,3,4 \}.$$ Let $\mathbf{i}$ be fixed. There are $k_{j_2}$ possible choices for $j_2$, where
	\begin{equation}\label{eqn::choices}
		k_{j_2}= \begin{cases}
			n -i_1 - 1, \textnormal{ if } i_3 < i_1\\
			n -i_1 - 2,  \textnormal{ else for } i_3 > i_1
		\end{cases}
	\end{equation}
	(since $j_2 > i_1$). There are $k_{j_3}=n-4$ choices for $j_3$, since $i_1$, $i_2$, $i_3$ and $j_2$ are excluded, i.e. already chosen. There are $k_{j_4}=m-1$ possibilities for $j_4$, since $i_4$ is excluded. 
	We abbreviate for a fixed $\mathbf{i}$ $$\{ i_1 = j_1\} :=\{\mathbf{j}\in I^{(n,m)}|  i_1 = j_1 \textnormal{ and } i_p \neq j_q \textnormal{ for } p,q =2,3,4\}.$$
	For a fixed $\mathbf{i}$ it is therefore
	\begin{align*}
		\#  \{ i_1 = j_1\} =1\cdot k_{j2}\cdot k_{j_{3}}\cdot k_{j_{4}}
		= (n-4)(m-1)\begin{cases}
			(n -i_1 - 1), \textnormal{ if } i_3 < i_1\\
			(n -i_1 - 2),  \textnormal{ if } i_3 > i_1
		\end{cases}.
	\end{align*}
	We obtain the desired cardinality $\# C_{i_1 = j_1}$ by summation over all $\mathbf{i} \in I^{(n,m)}$:
	\begin{align}\label{eqn::i=j}
		\# C_{i_1 = j_1}=\sum_{\mathbf{i} \in I^{(n,m)}} \# \{ i_1 = j_1\} &= (n-4)(m-1) (\sum_{\mathbf{i} \in I^{(n,m)}, i_1 > i_3}(n- i_1 - 1) + \sum_{\mathbf{i} \in I^{(n,m)}, i_1 < i_3} (n- i_1 - 2)) \notag\\
		&= (n-4)(m-1) \bigg( n \#  I^{(n,m)} - \# \{ \mathbf{i} \in I^{(n,m)}| i_1 > i_3\} \notag \\ 
		&\qquad- 2 \# \{ \mathbf{i} \in I^{(n,m)}| i_1 < i_3\} - \sum_{\mathbf{i} \in  I^{(n,m)} } i_1\bigg).
	\end{align}
	and since it is
	\begin{equation*}
		\# \{ \mathbf{i} \in I_1^{(n,m)}| i_1 < i_3\} +  \# \{ \mathbf{i} \in I^{(n,m)}| i_1 > i_3\} = \# I^{(n,m)}  
	\end{equation*}
	we have
	\begin{equation*}
		\# C_{i_1 = j_1}=(n-4)(m-1) \bigg( (n - 2) \#  I^{(n,m)} +   
		\# \{ \mathbf{i} \in I^{(n,m)}| i_1 > i_3\} - \sum_{\mathbf{i} \in  I^{(n,m)} } i_1\bigg)
	\end{equation*}
	For the terms in \eqref{eqn::i=j}, we find
	\begin{align*}
		\sum_{\mathbf{i} \in  I^{(n,m)} } i_1 &= \sum_{i_1 = 1}^n i_1 \# \{ \mathbf{k} \in I^{(n,m)}| k_1 = i_1\} 
		=  \sum_{i_1 = 1}^n i_1 (n-i_1)(n-2)m \\
		&=\frac{1}{2}n^2(n+1)(n-2)m-\frac{1}{6}n(n+1)(2n+1)m(n-2)\\
		&= \frac{1}{6}nm (n-1)(n-2)(n+1),
	\end{align*}
	as well as
	\begin{align*}
		\# \{ \mathbf{i} \in I^{(n,m)}| i_1 > i_3\} &=  \sum_{i_1 = 1}^n (i_1 - 1)(n-i_1)m 
		=m  \sum_{i_1 = 1}^n (ni_1 + i_1 -i_1^2 - n ) \\
		&=\frac{1}{2}nm(n+1)^2-\frac{1}{6}n(n+1)(2n+1)-mn^2
		= nm\left(\frac{1}{6}n^2 - \frac{1}{2}n +\frac{1}{3}\right),
	\end{align*}
	and $\# I^{(n,m)} = \frac{1}{2}n(n-1)(n-2)m$. Plugging these results into Eq. (\ref{eqn::i=j}), we obtain
	\begin{equation}\label{eqn::Ci1_j1}
		\# C_{i_1 = j_1} = \frac{1}{3}nm(m-1)(n-1)(n-2)(n-3)(n-4).
	\end{equation}

	\noindent The final results of all other cases discussed below will be expressed as multiples of $\# C_{i_1 = j_1}$.\\
	
	\paragraph{Case 2: $\mathbf{i_2 = j_2}$ results in $\mathbf{\# C_{i_2 = j_2} = \# C_{i_1 = j_1}}$.}~
	\smallskip
	
	\noindent Let $C_{i_2 = j_2} := \{ (\mathbf{i}, \mathbf{j})\in I^{(n,m)}\times I^{(n,m)} | i_2 = j_2 \textnormal{ and } i_p \neq j_q \textnormal{ for } p,q =1,3,4 \}$. In this case, for any fixed $i_2 = j_2$, there are $i_2 - 1$ and $i_2 - 2$ possibilities to choose $i_1$ and $j_1$, then $n-3$ and $n-4$ possible choices for $i_3$ and $j_3$ since $3$ and $4$ indices have already been chosen and fixed and there are $m$ and $m-1$ possibilities for $i_4$ and $j_4$. It is therefore
	\begin{align*}
		\#  C_{i_2 = j_2} &= \sum_{i_2 = 2}^n (i_2 - 1)(i_2 - 2)(n-3)(n-4) m (m-1) \\
		&= m(m-1) (n-3) (n-4) \sum_{i_2 = 2}^n (i_2^2 -3i_2 + 2).
	\end{align*}
	Since 
	\begin{align*}
		\sum_{i_2 = 2}^n (i_2^2 -3i_2 + 2) &= \frac{1}{6}n(n+1)(2n+1) - 1 -3(\frac{1}{2}n(n+1) - 1) +2(n-1) \\
		&= \frac{1}{3}n(n-1)(n-2)
	\end{align*}
	we conclude $\# C_{i_2 = j_2} = \# C_{i_1 = j_1}$.\\

	\paragraph{Case 3: $\mathbf{i_2 = j_1}$ results in $\mathbf{ \# C_{i_2 = j_1} = \frac{1}{2}\#  C_{i_1 = j_1} }$.}~\smallskip
	
	\noindent Similarly, for $C_{i_2 = j_1} := \{ (\mathbf{i}, \mathbf{j})\in I^{(n,m)}\times I^{(n,m)} | i_2 = j_1 \textnormal{ and } i_p \neq j_q \textnormal{ for } p =1,3,4, q = 2,3,4 \}$ we have, for a fixed $i_2 = j_1$,
	$i_2 - 1$ possibilities for the choice of $i_1$, $n-i_2$ possibilities for the choice of $j_2$ since $i_2 = j_1 < j_2 $, $n-3$ and $n-4$ for $i_3$ and $j_3$ and $m$ and $m-1$ for $i_4$ and $j_4$ and therefore 
	\begin{align*}
		\#  C_{i_2 = j_1} &= \sum_{i_2 = 2}^n (i_2 - 1)(n-i_2)(n-3)(n-4) m (m-1) \\
		&= m(m-1) (n-3) (n-4) \sum_{i_2 = 2}^n (ni_2 - n -i_2^2 +i_2)
	\end{align*}
	Since the latter sum equals $\frac{1}{6}n(n-1)(n-2)$, we conclude $ \# C_{i_2 = j_1} = \frac{1}{2}\#  C_{i_1 = j_1}$.\\
	
	\paragraph{Cases 4 and 5: $\mathbf{i_3 = j_1}$ and $\mathbf{i_3 = j_2}$ result in $\mathbf{C_{i_3 = j_1} =C_{i_3 = j_2}= \frac{3}{4} \#  C_{i_1 = j_1}}$.}~
	\smallskip
	
	\noindent For a fixed $i_3 = j_1$ there are $\binom{n-1}{2}$ choices for $i_1$ and $i_2$, $k_{j_2}$ choices for $j_2$ where
	\begin{equation*}
		k_{j_2} = \begin{cases}
			n -i_3, \textnormal{ if } i_1 < i_2 < i_3\\
			n -i_3 -1,  \textnormal{ if } i_1 < i_3 < i_2 \\
			n -i_3 -2,  \textnormal{ if } i_3 < i_1 < i_2
		\end{cases}.
	\end{equation*}
	and furthermore $n-4$ choices for $j_3,$ $m$ choices for $i_4$ and $m-1$ choices for $j_4$ and therefore
	\begin{align*}
		\#  C_{i_3 = j_1} &=(n-4)m(m-1)\left(  \frac{1}{2}(n-1)(n-2)\sum_{i_3 = 1}^n (n-i_3) -3\binom{n}{3}\right),
	\end{align*}
	where $\binom{n}{3}$ is the number of different possibilities to pick indices $i_1, i_2$ and $i_3$ between $1$ and $n$ such that $i_1 < i_3 < i_2$. The same for the ordering $i_3 < i_1 < i_2$. This yields
	\begin{align*}
		\#  C_{i_3 = j_1} &=(n-4)m(m-1)\left(  \frac{1}{2}(n-1)(n-2)\bigg(n^2-\frac{n(n+1)}{2}\bigg)-\frac{n(n-1)(n-2)}{2}\right)\\
		&=(n-4)m(m-1)\left(  \frac{1}{4}n(n-1)(n-2)(n-3)\right),
	\end{align*}

	and hence $\#  C_{i_3 = j_1} = \frac{3}{4} \#  C_{i_1 = j_1}$. Since it holds $\#  C_{i_3 = j_1} = \#  C_{i_3 = j_2} $, we also obtain $\#  C_{i_3 = j_2} = \frac{3}{4} \#  C_{i_1 = j_1}$.\\
	
	\paragraph{Case 6: $\mathbf{i_3 = j_3}$ results in $\mathbf{\#  C_{i_3 = j_3} = \frac{3}{4} \#  C_{i_1 = j_1}}$.}~
	\smallskip
	
	\noindent Let $i_3 = j_3$ be fixed. Choosing a $2$ element subset of the rest of $n-1$ points defines $i_1$ and $i_2$, i.e. there are $\binom{n-1}{2}$ possibilities and similarly, there are $\binom{n-3}{2}$ possibilities to chose $j_1$ and $j_2$ from the rest $n-3$ of points, while $i_4$ and $j_4$ can be chosen in $m$ and $m-1$ ways, therefore having
	\begin{align*}
		\#  C_{i_3 = j_3} &= \sum_{i_3=1}^n \frac{1}{2}(n-1)(n-2)\frac{1}{2} (n-3)(n-4) m(m-1) \\
		&= \frac{1}{4} nm(m-1)(n-1)(n-2)(n-3)(n-4)
	\end{align*}
	and we conclude $\#  C_{i_3 = j_3} = \frac{3}{4} \#  C_{i_1 = j_1}$.\\
	
	\paragraph{Case 7: $\mathbf{i_4 = j_4}$ results in $\mathbf{\#  C_{i_4 = j_4} =\frac{3}{4} \frac{n-5}{m-1}   \#  C_{i_1 = j_1}}$.}~
	\smallskip
	
	\noindent For fixed $i_4 = j_4$, there are $\binom{n}{2}$ choices of $i_1$ and $i_2$ and $\binom{n-2}{2}$ choices for $j_1$ and $j_2$, $n-4$ and $n-5$ for $i_3$ and $j_3$ therefore having
	\begin{align*}
		\#  C_{i_4 = j_4} &= \sum_{i_4=1}^m \frac{1}{2}n(n-1)\frac{1}{2} (n-2)(n-3) (n-4)(n-5) \\
		&= \frac{1}{4} nm(n-1)(n-2)(n-3)(n-4)(n-5) \\
		&= \frac{3}{4} \frac{n-5}{m-1}   \#  C_{i_1 = j_1} .
	\end{align*}
	
	\paragraph{Combining the findings of \textbf{Case 1} - \textbf{Case 7}}~
	\smallskip
	\noindent
	For the rest of the cases in Table \ref{tbl::ij} the combinatorial symmetry holds. We therefore have $\# C_{i_2 = j_2} = \# C_{i_1 = j_1}$, $\# C_{i_1 = j_2} = \# C_{i_2 = j_1} = \frac{1}{2} \# C_{i_1 = j_1} $, $\# C_{i_3 = j_1} = \# C_{i_3 = j_2} = \# C_{i_1 = j_3} = \# C_{i_2 = j_3} = \# C_{i_3 = j_3} = \frac{3}{4} \# C_{i_1 = j_1}$ and $\# C_{i_4 = j_4} = \frac{3}{4} \frac{n-5}{m-1}\# C_{i_1 = j_1}$, as summarized in Table \ref{tbl::ijcard} below, where all  cardinalities are listed as multiples of the reference cardinality $\#C_{i_1=j_1}$.
	
	\begin{table}[h]
		\begin{center}
			\begin{tabular}{|c|c|c|c|c|}
				\hline
				$ \mathbf{j} \setminus \mathbf{i}$ & $i_1$ & $i_2$ & $i_3$ & $i_4$ \\
				\hline
				$j_1$ & \cellcolor{green!25}1& \cellcolor{green!25} $\frac{1}{2}$ & \cellcolor{red!25} $\frac{3}{4}$ &  \cellcolor{gray!25}$0$ \\
				\hline
				$j_2$&  \cellcolor{green!25}$\frac{1}{2}$ &\cellcolor{green!25}  $1$ &  \cellcolor{red!25} $\frac{3}{4}$ & \cellcolor{gray!25}$0$\\
				\hline
				$j_3$&  \cellcolor{red!25}$\frac{3}{4}$ & \cellcolor{red!25}  $\frac{3}{4}$  &  \cellcolor{green!25} $\frac{3}{4} $ &   \cellcolor{gray!25}$0$ \\
				\hline
				$j_4$&  \cellcolor{gray!25}$0$ &  \cellcolor{gray!25} $0$ &  \cellcolor{gray!25} $0$ &\cellcolor{green!25} $\frac{3}{4}\frac{n-5}{m-1}$ \\
				\hline
			\end{tabular}
			
		\end{center}
		\caption{Cardinalities of  configurations reported in Table \ref{tbl::ij} in matching colors as multiples of the reference cardinality $\#C_{i_1=j_1}$.}
		\label{tbl::ijcard}
	\end{table}
	
	Since $2P_1-1 = -(2(1-P_1) - 1)$ the corresponding summands in the sum (\ref{eqn::VarT1}) are negative and we conclude
	\begin{align*}
		\frac{1}{4}    \underset{\# \mathbf{i} \cap \mathbf{j} = 1}{\sum_{(\mathbf{i}, \mathbf{j}) \in E_{n,m}}} (2\Prob(\I_{\mathbf{i}}= \I_{\mathbf{j}}) - 1) & =  \frac{1}{4} \# C_{i_1 = j_1} \bigg(2\cdot1+2\cdot\frac{1}{2}+\frac{3}{4} +\frac{3}{4}\frac{n-5}{m-1}-4\cdot\frac{3}{4} \bigg) (2 P_1 -1) \\
		& = \frac{3}{16} \frac{n+m-6}{m-1}   \# C_{i_1 = j_1} (2P_1 - 1).
	\end{align*}
	By \eqref{eqn::VarT1} and the negligibility of the diagonal elements and of the sums over pairs $(\mathbf{i}, \mathbf{j})$ sharing $2$ and $3$ indices, the claim of the theorem follows. 
	
	\hfill$\Box$

	\subsection{Proof Theorem \ref{thm::P2}}
	\noindent
	We first prove that $P_1<\frac{2}{3}$. to this end, let
	$P_2 := \Prob( \I_{\mathbf{i}} = \I_{\mathbf{j}})$ where $\mathbf{i}$ and $\mathbf{j}$ share exactly the first two indices, i.e.
	\begin{equation*}
		P_2 = \Prob (\I_{ \{d(X_1, X_2) \leq d(X_3, X_4) \}} = \I_{ \{d(X_1, X_2) \leq d(X_5, X_6) \}}).
	\end{equation*}
	We now show that $P_2=\frac{2}{3}$ and that $P_1<P_2$.\\
	Since $X_1$, ..., $X_6$ are i.i.d., so are the random variables $A:= d(X_1, X_2)$, $B:=d(X_3, X_4)$ and $C:=d(X_5, X_6)$. We rewrite $P_2$ as 
	$
	P_2=\mathbb{P}(A\leq B\wedge A\leq C).
	$
	It is therefore 
	\begin{align*}
		P_2 &= 1 - \Prob((A\leq B \land A>C) \lor (A\leq C \land A>B) ) \\
		&= 1-\left(\Prob(A\leq B \land A>C)+P(A\leq C \land A>B)\right) \\
		&= 1-2 \Prob(C<A\leq B),
	\end{align*}
	having $\Prob (A \leq B \land A > C) = \Prob (A > B \land A \leq C) = \Prob(C<A\leq B)$. If $\Prob(A=B)= 0$, i.e. ties almost surely do not occur,  $C < A \leq B$  is one of $6$ equally probable orderings of three distances and we conclude 
	\begin{equation*}
		P_2 = \frac{2}{3},
	\end{equation*}
	for any underlying distribution.\\ 
	\noindent
	Recall the definition of $P_1$ under null hypothesis:
	\begin{equation*}
		P_1 = \Prob (\I_{ \{d(X_1, X_2) \leq d(X_3, X_4) \}} = \I_{ \{d(X_1, X_7) \leq d(X_5, X_6) \}}).
	\end{equation*}
	Denoting by $D:= d(X_1, X_7)$ we similarly have 
	\begin{equation*}
		P_1 = 1 - 2\Prob(A \leq B \land C < D)
	\end{equation*}
	where $A$, $B$, $C$ and $D$ are identically distributed but $A$ and $D$ stochastically dependent. Since any permutation of the tupel $(A,B,C,D)$ swapping $A$ and $D$ or $B$ and $C$ retains its distribution  (e.g.\ $(A,B,C,D) \overset{\mathcal{D}}{=} (D,B,C,A)$), we have
	\begin{equation} \label{eqn::probAD1}
		\Prob (A < C < D < B) = \Prob (D < C < A < B)
	\end{equation}
	and 
	\begin{equation} \label{eqn::probAD2}
		\Prob (A < B < C < D) = \Prob (A < C < B < D).
	\end{equation}
	By decompositions
	\begin{equation*}
		\{ C < A \leq B\} = \{ D < C < A \leq B\} \dot\cup \{ C < D < A \leq B\} \dot\cup \{ C < A < D \leq B\} \dot\cup \{ C < A \leq B < D\} 
	\end{equation*}
	and 
	\begin{align*}
		\{ A \leq B \land C < D \} =  &\{ A < B < C < D \} \dot\cup \{A < C < B < D  \} \dot\cup \{ C < A < B < D \} \\ 
		\dot\cup &\{ C < A < D < B \} \dot\cup \{ C< D < A < B  \} \dot\cup \{ A < C < D < B \}
	\end{align*}
	we have
	\begin{align*}
		\Prob(A \leq B \land C < D) - \Prob(C<A\leq B) &= \Prob(A < B < C < D) + \Prob (A < C < B < D) \\
		&+ \Prob (A < C < D < B) - \Prob (D < C < A < B)\\
		&= 2 \Prob(A < B < C < D)
	\end{align*}
	Since $ \Prob(A < B < C < D) > 0$ we conclude that $\Prob(A \leq B \land C < D) >  \Prob(C<A\leq B)$  and therefore 
	\begin{equation*}
		P_1 < P_2=\frac{2}{3},
	\end{equation*}
	again, for any underlying distribution of $X$, which establishes the second inequality of \eqref{eqn::P1bound}. \\
	
	\noindent
	It remains to show the first inequality in (\ref{eqn::P1bound}). Let the first condition of the theorem hold, i.e. let the support $S$ of the density $f$ of the distribution $X$ be unbounded, with $\Prob (\|X\| > t) > 0$ for any $t$, and let $p(x):= \Prob(x \in B_{\|X_1 - X_2\|}(X_3))$, where $X_1$, $X_2$ and $X_3$ be three independent instances of $X$. By Example \ref{bsp::e1} it is then
	\begin{equation*}
		P_1 = 2 \E_X \left[p(X)^2 \right].
	\end{equation*}
	Let further $t$ be such that $\Prob (\|X\| > t) = \frac{1}{12}$. By the triangle and reversed triangle inequalities it is
	\begin{align*}
		p(x) &\leq \Prob( \|x\| \leq \|X_1\|+ \|X_2\| + \|X_3\|) \\
		& \leq \Prob( \{\|x\| \leq \|X_1\| + \|X_2\| + \|X_3\|\} \cap_{i=1}^3 \{ \|X_i\| \leq t \} ) + \Prob( \cup_{i=1}^3 \{ \|X_i\| > t \}) \\
		& \leq \I_{ \{\|x\| \leq 3t \}} + \sum_{i = 1}^3 \Prob(\|X_i\| > t ) \\
		& = \I_{ \{\|x\| \leq 3t \}} +\frac{1}{4}
	\end{align*}
	and we conclude that $p(x) \leq \frac{1}{4}$ if $\|x\| > 3t$. Since $y\mapsto1-2y+2y^2$ is monotonically decreasing with global minimum $\frac{1}{2}$ at $y=\frac{1}{2}$,   for such $x$ it is 
	\begin{equation*}
		1 - 2p(x) + 2p(x)^2 \geq \frac{5}{8}
	\end{equation*}
	and therefore
	\begin{align*}
		\E \left[ 1-2p(X) + 2p(X)^2 \right]  &\geq \E  \left[ (1-2p(X) + 2p(X)^2) (\I_{ \|X\|\leq 3t } + \I_{ \|X\| > 3t }) \right] \\
		& \geq \frac{1}{2} \Prob(\|X\| \leq 3t) + \frac{5}{8}  \Prob(\|X\| > 3t) \\
		& = \frac{1}{2} + \frac{1}{8} \Prob(\|X\|> 3t).
	\end{align*}
	Since $\mathbb{P}(\|X\|>3t)>0$ we conclude that  $P_1 > \frac{1}{2}$. \\
	
	\noindent
	Otherwise, let the second condition hold, i.e. let the distribution of $X$ be such that (\ref{eqn::conc}) holds. Note that $P_1 = \frac{1}{2}$ iff $p(x) = \frac{1}{2}$ almost surely. We have
	\begin{equation*}
		p_r(x) \leq p_r(m)
	\end{equation*}
	for some $m,x\in S_X$ and all $r \in S_R$, where the inequality is strict on a set $S_{0,R}$ with positive mass and therefore
	\begin{align*}
		p(x)  = \int_{\R^{+}} p_r(x) \rm{d} \Prob_R(r) < \int_{\R^{+}} p_r(m) \rm{d} \Prob_R(r) = p(m).  
	\end{align*}
	For the continuous function  $p$ this is a contradiction to $p(x) = \frac{1}{2}$ a.s. Hence we conclude $P_1 > \frac{1}{2}$.

	\hfill$\Box$
	
	\subsection{Central limit theorem for locally dependent variables}
	\noindent
	For the sake of completeness, in the following we formulate the used version of CLT for triangular arrays under local dependency (represented by a dependency graph) as developed in \cite{Austern22}. Suppose that $p\geq 1$ is a real number and $\omega := p + 1 -\lceil p \rceil \in \left[ 0,1 \right] $. $I$ is defined as an infinite index set and $(I_n)_{n \in \N}$ the accompanying sequence of finite subsets with $I_1 \subset I_2 \subset ... \subset I$ with $\# I_n \rightarrow \infty$. For a triangular array of the form $(Z_i^{(n)})_{i \in I_n}$ we define
	\begin{equation*}
		W_n := \sigma_n^{-1} \sum_{i \in I_n} Z_i^{(n)} 
	\end{equation*}
	where $\sigma_n^2 := \Var \left( \sum_{i \in I_n} Z_i^{(n)} \right)$. The triangular array induces the dependency graph $G_n = (I_n,E_n)$ with 
	\begin{equation*}
		(i,j) \in E_n \iff Z_i^{(n)} \textnormal{ and } Z_j^{(n)} \textnormal{ dependent.}
	\end{equation*}
	
	\noindent
	The necessary condition on the dependence structure holding for all $n \in \N$ is defined as the following condition:\\
	
	\noindent
	[LD*] There exist the dependency graph $G_n = (I_n,E_n)$ s.t. for all disjoint subsets $J_1, J_2 \subset I_n$ if there are no edges between $J_1$ and $J_2$, then $\{Z_j^{(n)} | j \in J_1 \}$ is independent of $\{Z_j^{(n)} | j \in J_2 \}$.\\
	
	\noindent
	The set of vertices in the neighborhood of $J \subset I_n$ is defined as 
	\begin{equation*}
		N_n(J) := J \cup \{ i \in I_n | (i,j)\in E_n \textnormal{ for some } j\in J \}
	\end{equation*}
	\noindent
	We can proceed to formulate the statement corresponding to the Theorem 3.3 in (\cite{Austern22}).
	\begin{Theorem} \label{thm::CLT}
		Suppose that $(Z_i^{(n)})$ is a triangular array of mean zero random variables satisfying [LD*], and that the cardinality of index set $N_n(\{ i_1,...,i_{\lceil p \rceil +1} \})$ is upper-bounded by $M_n < \infty$ for any $i_1,...,i_{\lceil p \rceil +1}  \in I_n$. Furthermore, assume that
		\begin{equation}\label{eqn::CLT34}
			M_n^{1+\omega} \sigma_n^{-(\omega + 2)} \sum_{i \in I_n} \E \left[ |Z_i^{(n)}|^{(\omega+2)} \right] \rightarrow 0 \textnormal{  and   } M_n^{1+p} \sigma_n^{-(p + 2)} \sum_{i \in I_n} \E \left[ |Z_i^{(n)}|^{p+2} \right] \rightarrow 0.
		\end{equation}
		Then there is $N$ such that for all $n \geq N$ we have
		\begin{align*}
			&\mathcal{W}_p(\mathcal{L}(W_n), \mathcal{N}(0,1)) \\ 
			&\leq C_p \left( M_n^{1+\omega} \sigma_n^{-(\omega + 2)} \sum_{i \in I_n} \E \left[ |Z_i^{(n)}|^{(\omega+2)} \right] \right)^{1/\omega} + C_p \left( M_n^{1+p} \sigma_n^{-(p + 2)} \sum_{i \in I_n} \E \left[ |Z_i^{(n)}|^{(p+2)} \right] \right)^{1/p}
		\end{align*}
		for some constant $C_p$ that only depends on $p$.
	\end{Theorem}

	\subsection{Proof of Theorem  \ref{thm::T}}
	\noindent
	Theorem \ref{thm::T} is proved by an application of Theorem \ref{thm::CLT}. Accordingly, we have to show that Asssumption [LD*] holds and that \eqref{eqn::CLT34} is satisfied. To verify  Asssumption [LD*], note that $I^{(n, m)}$ and $E^{(n,m)}$, defined by the relation $(\mathbf{i}, \mathbf{j}) \in E^{(n,m)}$ if and only if $\mathbf{i} \neq \mathbf{j}$ and  $\mathbf{i} \cap \mathbf{j} \neq \emptyset$, 
	induce an undirected  loopless  graph $G^{(n,m)} = (I^{(n,m)}, E^{(n,m)})$.
	Moreover,  $I^{(n_1, m_{n_1})} \subseteq I^{(n_2, m_{n_2})} \subseteq  \N^4$ for $n_2\geq n_1$, since by assumption $m=m_n$ is monotonically increasing in $n$.
	For two disjoint subsets $J_1$, $J_2\subset I^{(n, m)}$ no edges between $J_1$ and $J_2$ means that for any two indices $j_1$ entering $J_1$ and $j_2$ entering $J_2$ it holds that $j_1\neq j_2$. As a result, $\{\mathcal{Z}_{\mathbf{j}}| \mathbf{j}\in J_1\}$ and $\{\mathcal{Z}_{\mathbf{j}}| \mathbf{j}\in J_2\}$ are composed  of functions of  disjoint subsets of $\{X_1, \ldots, X_n, Y_{n+1}, \ldots, Y_{n+m}\}$. Since the latter set consists of mutually independent random variables, the set $\{\mathcal{Z}_{\mathbf{j}}| \mathbf{j}\in J_1\}$ is independent of $\{\mathcal{Z}_{\mathbf{j}}| \mathbf{j}\in J_2\}$.
	Therefore, for the  graph $G^{(n,m)}$ defined in the above manner,
	condition [LD*] of  Theorem \ref{thm::CLT} holds.\\
	\noindent
	For showing that condition \eqref{eqn::CLT34} of  Theorem \ref{thm::CLT} is fulfilled, recall that the statistic we are interested in is defined as follows:
	\begin{equation}\label{variance}
		W_{n,m} = \frac{1}{\sigma_{n,m}} \sum_{\mathbf{i} \in I^{(n,m)}} \mathcal{Z}_{\mathbf{i}}, \
		\text{where} \
		\mathcal{Z}_{\mathbf{i}}:= \I_{\mathbf{i}} - \mathbf{E}_1
	\end{equation}
	with  $\I_{\{\mathbf{i}}=\I_{d(X_{i_1}, X_{i_2}) \leq d(X_{i_3}, Y_{i_4})\}}$  for $\mathbf{i}=(i_1, i_2, i_3, i_4)$,  $\mathbf{E}_1 = \Prob(d(X_1, X_2)\leq d(X_3, Y_1))$, and
	$\sigma_{n,m}^2 = \Var \left(\sum_{\mathbf{i} \in I^{(n,m)} } \mathcal{Z}_{\mathbf{i}} \right)$.
	We show that condition \eqref{eqn::CLT34}  is fulfilled for  the above statistic when $p = \omega = 1$. For this, we first determine  a lower bound for $\sigma_{n, m}$. 
	Note that 
	\begin{align*}
		\sigma_{n, m}^2 =
		\sum_{\mathbf{i} \in I^{(n,m)} }\sum_{\mathbf{j} \in I^{(n,m)} } \E\left(\mathcal{Z}_{\mathbf{i}}, \mathcal{Z}_{\mathbf{j}}\right).
	\end{align*}
	The sum on the right-hand side of the  above equation is asymptotically dominated by summands $ \E \left( \mathcal{Z}_{\mathbf{i}}\mathcal{Z}_{\mathbf{j}}\right)$ for which $\mathbf{i}$ and $\mathbf{j}$ share exactly one index.
	The rest of the summands, i.e. those where $\mathbf{i}$ and $\mathbf{j}$ share two, three, or four indices are asymptotically of lower order (cf. proof of Theorem \ref{thm::Var}), such that 
	\begin{align*}
		\sigma_{n, m}^2  \gtrsim    \underset{\# \mathbf{i} \cap \mathbf{j} = 1}{\sum_{\mathbf{i}, \mathbf{j} \in I^{(n,m)}}} \E\left(\mathcal{Z}_{\mathbf{i}}, \mathcal{Z}_{\mathbf{j}}\right).
	\end{align*}
	By definition of $\mathcal{Z}_{\mathbf{i}}$ and $\mathcal{Z}_{\mathbf{j}}$
	\begin{align*}
		\E \left( \mathcal{Z}_{\mathbf{i}}\mathcal{Z}_{\mathbf{j}}\right) = 
		\E \left( \I_{\mathbf{i}}\I_{\mathbf{j}}\right) -   \E \left( \I_{\mathbf{i}}\right)\E \left( \I_{\mathbf{j}}\right)
		=\Prob \left( \I_{\mathbf{i}}=\I_{\mathbf{j}}=1\right) -\mathbf{E}^2_1.
	\end{align*}
	Our goal is to determine the above expression explicitly. 
	For this, we express  $\E \left( \mathcal{Z}_{\mathbf{i}}\mathcal{Z}_{\mathbf{j}} \right)$
	in terms of the following quantities:
	\begin{align*}
		p(x)&:= \Prob (d(X_1, x)\leq d(X_2, Y_1)), \\
		q(x)&:= \Prob (d(Y_1, x) \leq d(X_1, X_2)), \\
		r(y)&:= \Prob (d(X_1, y) \leq d(X_2, X_3)),
	\end{align*}
	for which, due to the fact that $\Prob (d(X_1, X_2)= d(X_3, Y_1))=0$,
	\begin{equation*}
		\E_Y \left( r(Y) \right) =  \E_X \left(q(X) \right) = 1-\mathbf{E}_1
		\ \text{and} \
		\E_X \left( p(X) \right) = \mathbf{E}_1.
	\end{equation*} 
	\noindent
	
	As in the proof of Theorem \ref{thm::Var}, to determine $\E \left( \mathcal{Z}_{\mathbf{i}}\mathcal{Z}_{\mathbf{j}} \right)$ we make a case distinction with respect to which elements of $\mathbf{i}=(i_1, i_2, i_3, i_4)$ and  $\mathbf{j}=(j_1, j_2, j_3, j_4)$  coincide.
	The argument for the individual cases is analogous to  the proof of Lemma \ref{lemma::example1}.l
	\\
	\\
	\noindent
	\textbf{Case 1: $i_1=j_1$}\\
	Note that 
	\begin{align*}
		\Prob (\I_{\mathbf{i}}=\I_{\mathbf{j}}=1)=\E_X\left( h(X)\right), \ h(x)= \Prob(Z_1(x) =1,  Z_2(x)=1),
	\end{align*}
	where $Z_1(x):= \I_{\{d(x,X_1)\leq d(X_2, Y_1)\}} $ and $Z_2(x):= \I_{\{d(x,X_3)\leq d(X_4, Y_2)\}}$, i.e. $Z_1(x)$ and $Z_2(x)$ are independent, identically distributed random variables with $P(Z_1(x)=1)=P(Z_2(x)=1)=p(x)$.
	It follows that   $\Prob (\I_{\mathbf{i}}=\I_{\mathbf{j}}=1)=\E_X\left(p^2(X)\right)$.
	Therefore, we have
	\begin{align*}
		\E \left[ \mathcal{Z}_{\mathbf{i}}\mathcal{Z}_{\mathbf{j}}\right] 
		&=  \E_X\left(p^2(X)\right)-\mathbf{E}_1^2.
	\end{align*}
	\textbf{Case 2: $i_2=j_2$, Case 3: $i_1=j_2$, Case 4: $i_2=j_1$ }\\
	Due to symmetry of $d$ and interchangeability of $\mathbf{i}$ and $\mathbf{j}$, the same argument as in \textbf{Case 1}
	yields      $\E \left( \mathcal{Z}_{\mathbf{i}}\mathcal{Z}_{\mathbf{j}}\right) = \E_X\left(p^2(X)\right)-\mathbf{E}_1^2$.\\
	\\
	\textbf{Case 5: $i_1=j_3$}\\
	Note that 
	\begin{align*}
		\Prob (\I_{\mathbf{i}}=\I_{\mathbf{j}}=1)=\E_X \left(h(X)\right), \    h(x)= \Prob(Z_1(x) =1,  Z_2(x)=1),
	\end{align*}
	where $Z_1(x):= \I_{\{d(x,X_1)\leq d(X_2, Y_1)\}} $ and $Z_2(x):= \I_{\{d(X_3, X_4)\leq d(x, Y_2)\}}$, i.e. $Z_1(x)$ and $Z_2(x)$ are independent random variables with $P(Z_1(x)=1)=p(x)$ and  $P(Z_2(x)=1)=1-q(x)$.
	It follows that   $\Prob (\I_{\mathbf{i}}=\I_{\mathbf{j}}=1)=\E_X\left(p(X)(1-q(X))\right)$.
	Therefore, we have
	\begin{align*}
		\E \left( \mathcal{Z}_{\mathbf{i}}\mathcal{Z}_{\mathbf{j}} \right)
		\ignore{   &=\left(1-\mathbf{E}_1\right) \Prob (\I_{\mathbf{i}}=\I_{\mathbf{j}}=1)+\mathbf{E}_1\Prob (\I_{\mathbf{i}}=\I_{\mathbf{j}}=0)-\mathbf{E}_1(1-\mathbf{E}_1))\\
			&=\left(1-\mathbf{E}_1\right)(\mathbf{E}_1-\E_X(p(X)q(X)))+\mathbf{E}_1(1-\mathbf{E}_1-\E_X(p(X)q(X)))-\mathbf{E}_1(1-\mathbf{E}_1)\\}
		=\E_X\left(p(X)(1-q(X))\right)- \mathbf{E}^2_1.
		\ignore{       =\mathbf{E}_1(1-\mathbf{E}_1)-\E_X(p(X)q(X)).}
	\end{align*}
	\\
	\textbf{Case 6: $i_2=j_3$, Case 7: $i_3=j_1$, Case 8:  $i_3=j_2$ }\\
	Due to symmetry of $d$ and interchangeability of $\mathbf{i}$ and $\mathbf{j}$, the same argument as in  \textbf{Case 5}  yields $\E \left( \mathcal{Z}_{\mathbf{i}}\mathcal{Z}_{\mathbf{j}}\right)=\E_X\left(p(X)(1-q(X))\right)- \mathbf{E}^2_1$. \\
	\noindent
	\textbf{Case 9: $i_3=j_3$}\\
	Note that 
	\begin{align*}
		\Prob (\I_{\mathbf{i}}=\I_{\mathbf{j}}=1)=\E_X \left(h(X)\right), \    h(x)= \Prob(Z_1(x) =1,  Z_2(x)=1),
	\end{align*}
	where $Z_1(x):= \I_{\{d(X_1,X_2)\leq d(x, Y_1)\}} $ and $Z_2(x):= \I_{\{d(X_3,X_4)\leq d(x, Y_2)\}}$,
	i.e. $Z_1(x)$ and $Z_2(x)$ are independent, identically distributed random variables with $P(Z_1(x)=1)=P(Z_2(x)=1)=1-q(x)$.
	It follows that   $\Prob (\I_{\mathbf{i}}=\I_{\mathbf{j}}=1)= \E_X (1-q(X))^2$.
	\ignore{= 1-2\E_X(q(X))+\E_X\left[(q(X))^2\right]
		= 1-2(1-\mathbf{E}_1)+\E_X\left[(q(X))^2\right]$.}
	\ignore{and similarly 
		\begin{equation*}
			\Prob (\I_{\mathbf{i}}=\I_{\mathbf{j}}=0) =\E_X \Prob(Z_1(X) = Z_2(X) = 0) = \E_X q(X)^2.
	\end{equation*}}
	Therefore, we have
	\begin{align*}
		\E \left( \mathcal{Z}_{\mathbf{i}}\mathcal{Z}_{\mathbf{j}}\right)
		=\E_X\left( (1-q(X))^2\right)-\mathbf{E}^2_1.
		\ignore{ &= \left(1-\mathbf{E}_1\right) \Prob (\I_{\mathbf{i}}=\I_{\mathbf{j}}=1)+\mathbf{E}_1\Prob (\I_{\mathbf{i}}=\I_{\mathbf{j}}=0)-\mathbf{E}_1(1-\mathbf{E}_1))\\
			&=\left(1-\mathbf{E}_1\right)(1 - 2(1 - \mathbf{E}_1) +\E_X (q(X)^2))+\mathbf{E}_1\E_X (q(X)^2)-\mathbf{E}_1(1-\mathbf{E}_1)\\
			&= 1-2(1-\mathbf{E}_1)+\E_X\left[(q(X))^2\right]- (\mathbf{E}_1)^2\\
			&=\E_X (q(X)^2) -(1 - \mathbf{E}_1 )^2}
	\end{align*}
	\\
	\textbf{Case 10: $i_4=j_4$} \\
	Note that 
	\begin{align*}
		\Prob (\I_{\mathbf{i}}=\I_{\mathbf{j}}=1)=\E_Y \left(h(Y)\right), \    h(y)= \Prob(Z_1(y) =1,  Z_2(y)=1),
	\end{align*}
	where
	$Z_1(y):= \I_{\{d(X_1,X_2)\leq d(X_3, y)\}} $ and $Z_2(y):= \I_{\{d(X_4,X_5)\leq d(X_6, y)\}} $, i.e.  $Z_1(y)$ and $Z_2(y)$ are independent, identically distributed random variables with $P(Z_1(y)=1)=  \Prob(Z_2(y) = 1) = 1- r(y)$.
	\ignore{and
		\begin{equation*}
			\Prob (\I_{\mathbf{i}}=\I_{\mathbf{j}}=0) =\E_Y \Prob(Z_1(Y) = Z_2(Y) = 0) = \E_Y\left[r(Y)^2\right]
		\end{equation*}
		and therefore}
	It follows that $ \Prob (\I_{\mathbf{i}}=\I_{\mathbf{j}}=1)=\E_Y\left((1-r(Y))^2\right)$.
	Therefore, we have
	\begin{align*}
		\E \left( \mathcal{Z}_{\mathbf{i}}\mathcal{Z}_{\mathbf{j}}\right)
		\ignore{      &=\left(1-\mathbf{E}_1\right) \Prob (\I_{\mathbf{i}}=\I_{\mathbf{j}}=1)+\mathbf{E}_1\Prob (\I_{\mathbf{i}}=\I_{\mathbf{j}}=0)-\mathbf{E}_1(1-\mathbf{E}_1))\\
			&=\left(1-\mathbf{E}_1\right)(1 - 2(1 - \mathbf{E}_1) +\E_Y (r(Y)^2))+\mathbf{E}_1\E_Y (r(Y)^2)-\mathbf{E}_1(1-\mathbf{E}_1) \\}
		&=  \E_Y \left((1-r(Y))^2\right)  - \mathbf{E}^2_1.
		\ignore{       &=\E_Y (r(Y)^2) -(1 - \mathbf{E}_1 )^2.}
	\end{align*}
	
	\noindent
	The number of summands of each case contributing to the variance reported in \eqref{variance} can be determined analogously as in the proof of Theorem \ref{thm::Var}. Table \ref{tbl::ijcard} represents all the possible configurations of $\mathbf{i}$ and $\mathbf{j}$ sharing exactly one index and the corresponding number of summands as  multiplicities of the number of summands corresponding to \textbf{Case 1} ($i_1 = j_1$), i.e. $\# C_{i_1=j_1}$, where
	$$C_{i_1 = j_1} := \{ (\mathbf{i}, \mathbf{j})\in I^{(n,m)}\times I^{(n,m)} | i_1 = j_1 \textnormal{ and } i_p \neq j_q \textnormal{ for } p,q =2,3,4 \}.$$

	\begin{table}[h]
		\begin{center}
			\begin{tabular}{|c|c|c|c|c|}
				\hline
				$ \mathbf{j} \setminus \mathbf{i}$ & $i_1$ & $i_2$ & $i_3$ & $i_4$ \\
				\hline
				$j_1$ &1&  $\frac{1}{2}$ & $\frac{3}{4}$ & $0$ \\
				\hline
				$j_2$&  $\frac{1}{2}$ & $1$ & $\frac{3}{4}$ &$0$\\
				\hline
				$j_3$& $\frac{3}{4}$ &  $\frac{3}{4}$  & $\frac{3}{4} $ &  $0$ \\
				\hline
				$j_4$& $0$ &  $0$ &  $0$ & $\frac{3}{4}\frac{n-5}{m-1}$ \\
				\hline
			\end{tabular}
		\end{center}
		\caption{Cardinalities of  configurations reported in Table \ref{tbl::ij} as multiples of the reference cardinality $\#C_{i_1=j_1}$.}
		\label{tbl::ijcard}
	\end{table}
	\noindent
	It follows that
	\begin{align*}
		\underset{\# \mathbf{i} \cap \mathbf{j} = 1}{\sum_{(\mathbf{i}, \mathbf{j}) \in E_{n,m}}}  \E \left(\mathcal{Z}_{\mathbf{i}}\mathcal{Z}_{\mathbf{j}} \right) 
		=&  \# C_{i_1 = j_1} \bigg(3\left(\E_X \left( p^2(X)\right) -\mathbf{E}_1^2\right)+ 
		3\left(\E_X\left(p(X)(1-q(X))\right)
		- \mathbf{E}^2_1\right) \\
		&+ \frac{3}{4}\left(\E_X\left( (1-q(X))^2\right)-\mathbf{E}^2_1\right) + \frac{3}{4} \frac{n-5}{m-1}(\E_Y \left((1-r(Y))^2\right)  - \mathbf{E}^2_1) \bigg) \\
		=&  \# C_{i_1 = j_1} \bigg( 3\E_X \left( p^2(X)\right) -3\E_X \left(p(X)q(X) \right) + \frac{3}{4} \E_X \left( q^2(X)\right) \\
		&+  \frac{3}{4} \frac{n-5}{m-1}\left(\E_Y \left(r^2(Y)\right) +2\mathbf{E}_1-\mathbf{E}^2_1-1\right)
		-\frac{27}{4}\mathbf{E}^2_1+\frac{9}{2}\mathbf{E}_1-\frac{3}{4}
		\bigg) \\ 
		\ignore{      =&  \# C_{i_1 = j_1}\bigg( 3\E_X \left( p^2(X)\right) -3\E_X \left(p(X)q(X) \right) + \frac{3}{4} \E_X \left( q^2(X)\right) +\frac{3}{4}  \frac{n-5}{m-1}\E_Y \left(r^2(Y)\right)  \\
			& +3\mathbf{E}_1(1-\mathbf{E}_1) -3 \mathbf{E}_1^2-\frac{3}{4} \left(1-\mathbf{E}_1\right)^2\left(1 + \frac{n-5}{m-1}\right) \bigg),}
	\end{align*}
	since $\E_Y \left( r(Y) \right) =  \E_X \left(q(X) \right) = 1-\mathbf{E}_1$
	and $\E_X \left( p(X) \right) = \mathbf{E}_1$.
	Since $\frac{n-5}{m-1} \underset{n,m \rightarrow \infty}{\rightarrow} \frac{\lambda}{1- \lambda}$ by assumption, it holds that
	\begin{align*}
		\underset{\# \mathbf{i} \cap \mathbf{j} = 1}{\sum_{(\mathbf{i}, \mathbf{j}) \in E_{n,m}}}  \E \left( \mathcal{Z}_{\mathbf{i}}\mathcal{Z}_{\mathbf{j}} \right)
		>\ignore{& \frac{3}{4} C \# C_{i_1 = j_1}\bigg( \E_X\left( (2p(X) - q(X))^2\right) + \frac{\lambda}{1- \lambda} \E_Y \left(r^2(Y)\right)  \\
			& + 4\mathbf{E}_1(1-\mathbf{E}_1) -4 \mathbf{E}_1^2- \frac{1}{1- \lambda}(1-\mathbf{E}_1)^2 \bigg)\\
			=& \frac{3}{4} C \# C_{i_1 = j_1} \frac{1}{1-\lambda} \bigg( (1-\lambda)\E_X\left( (2p(X) - q(X))^2\right)  + \lambda \E_Y \left( r^2(Y)\right) \\ 
			&- (1-\lambda) (4\mathbf{E}_1(1-\mathbf{E}_1) -4\mathbf{E}_1^2 -(1-\mathbf{E}_1)^2) -\lambda (1-\mathbf{E}_1)^2   \bigg) \\}
		& \frac{3}{4} C \# C_{i_1 = j_1} \frac{1}{1-\lambda} \bigg( (1-\lambda)\E_X\left((2p(X) - q(X))^2\right)  + \lambda \E_Y \left( r^2(Y)\right) \\ 
		&-(1-\lambda)(3\mathbf{E}_1 - 1 )^2 - \lambda (1-\mathbf{E}_1)^2 \bigg).
	\end{align*}
	for some positive constant $C$.
	
	Noting that $\E_X\left(2p(X) - q(X)\right) = 3\mathbf{E}_1 - 1 $ and $\E_Y\left( r(Y)\right) = 1-\mathbf{E}_1$, we conclude
	\begin{align}
		\underset{\# \mathbf{i} \cap \mathbf{j} = 1}{\sum_{(\mathbf{i}, \mathbf{j}) \in E_{n,m}}}  \E \left(\mathcal{Z}_{\mathbf{i}}\mathcal{Z}_{\mathbf{j}} \right)  &>  \frac{3}{4} C \# C_{i_1 = j_1} \frac{1}{1-\lambda} \bigg( (1-\lambda) \Var (2p(X) - q(X)) + \lambda \Var(r(Y)) \bigg)  > 0. \label{eqn::sigma_estimate}
	\end{align}
	The right-hand side of the above inequality is positive if $\Var(r(Y))>0$.
	For  $\Var(r(Y))>0$ it suffices that $r(Y)$  is not almost surely constant.
	Overall, it follows that 
	\begin{align}
		\sigma_{n,m}^2     
		>  \# C_{i_1 = j_1} C      \label{eqn::Varbound}
	\end{align}
	where $C$ is a positive constant that does not depend on $n$ or $m$. \\ 
	\noindent
	In order to verify condition (\ref{eqn::CLT34}) of Theorem \ref{thm::CLT}
	we have to establish an upper-bound $M_n$ on the cardinality of  dependency neighbourhoods $N(\{\mathbf{i}, \mathbf{j}\})$  for any two vertices $\mathbf{i}$, $\mathbf{j}$ of the  graph $(I^{n, m}, E^{(n,m)})$.
	We have
	\begin{align}
		\# N_{n,m}(\mathbf{i}, \mathbf{j}) & \leq \# N_{n,m}(\mathbf{i}) + \# N_{n,m}(\mathbf{j}) \notag \\
		& = 2  N_{n,m}, \label{eqn::Nnm}
	\end{align} 
	where 
	$N_{n,m}$ denotes the number of elements of any dependency neighbourhood of order one, since due to definition of the edge set $E^{(n, m)}$ all vertices possess the same number of adjacent vertices (c.f. Combinatorial Quantities \ref{subsec::combinatorial_quantities}).\\
	Combining  \eqref{eqn::Varbound} and \eqref{eqn::Nnm},  for proving validity of (\ref{eqn::CLT34})  we conclude that 
	\begin{align*}
		\frac{(\# N_{n,m}(\mathbf{i}, \mathbf{j}))^2  }{\sigma_{n,m}^3} \sum_{\mathbf{i} \in I^{(n,m)}} \E |\mathcal{Z}_{\mathbf{i}}|^3  &\leq \mathcal{C}  \frac{ N_{n,m}^2 \# I^{(n,m)}}{ (\# C_{i_1 = j_1})^{\frac{3}{2}}},
	\end{align*}
	where $\mathcal{C}$ is a constant not depending on the sample sizes $n$ and $m$. By \eqref{eqn::NoI} and \eqref{eqn::NoN1} in Section \ref{subsec::combinatorial_quantities}, we have 
	\begin{align*}
		\# I^{(n,m)} & = \frac{1}{2} n m (n-1) (n-2),\\
		N_{n,m} & = \frac{n^3}{2} + \frac{9m n^2}{2} - 6 n^2  -\frac{45mn}{2} +\frac{47n}{2} + 30 m  -30,
	\end{align*}
	i.e. $\# I^{(n,m)}   \sim n^3 m$ and $N_{n,m}   \sim n^3 + n^2m$.  By (\ref{eqn::Ci1_j1}) in Section \ref{sec:Tconv}, it holds that 
	\begin{equation*}
		\# C_{i_1 = j_1} = \frac{2}{3} \# I^{(n,m)} (m-1)(n-3)(n-4)
	\end{equation*}
	i.e. $\# C_{i_1 = j_1}  \sim n^5 m^2$ and hence, for $n$ and $m$ large enough and some constant $\mathcal{C}$ we have
	\begin{align*}
		\frac{(\# N_{n,m}(\mathbf{i}, \mathbf{j}))^2  }{\sigma_{n,m}^3} \sum_{\mathbf{i} \in I^{(n,m)}} \E |\mathcal{Z}_{\mathbf{i}}|^3  \leq \mathcal{C}  \frac{(n^3 + n^2m)^2}{n^{9/2} m^2} 
		\leq n^{-1/2} \mathcal{C}  \left(\frac{n+m}{m}\right)^2.
	\end{align*}
	Since $\frac{n+m}{m} \rightarrow \frac{1}{1-\lambda} \in (0, \infty) $ as $n,m \rightarrow \infty$, 
	condition \eqref{eqn::CLT34} of  Theorem \ref{thm::CLT} is fulfilled for $p = \omega = 1$. As a result,  the distribution of $W_{n,m}$ converges to a  normal distribution in Wasserstein-1 distance with convergence rate $n^{-1/2}$. \\
	\noindent
	\ignore{In particular, $W_{n,m}$ converges under the null hypothesis. Specifically, under the null hypothesis it is $\mathbf{E}_1 = \frac{1}{2}$ and $p=q=r$ with $p(x)=  \Prob (X_1 \in \mathcal{B}_{||X_2 - X_3||}(x))$ where $X_1$, $X_2$ and $X_3$ are three independent instances of $X$. Since due to Lemma \ref{lemma::example1} we have $P_1 = \E_X \left[1 - 2p(X) + 2p(X)^2 \right]$ it follows
		\begin{align*}
			\Var (p(X)) &= \frac{1}{2} (P_1 + 2 \E_X p(X) - 2(\E_X p(X))^2  -1) \\ 
			&=\frac{1}{2} (P_1 -\frac{1}{2}).
		\end{align*}
		Conditions of Theorem \ref{thm::P2} guarantee that $P_1 > \frac{1}{2}$ and the estimate (\ref{eqn::sigma_estimate}) holds, i.e. the order of the variance bound under the null hypothesis remains the same as in the proof of Theorem \ref{thm::T}. The Corollary \ref{cor::folgerung1} follows.\\
		In particular, $p(X)$ almost surely constant, i.e. $p(X)= \mathbf{E}_1 $, is equivalent to $q(X)$ and $r(Y)$ almost surely constant with $q(X)=r(Y)= 1-  \mathbf{E}_1$ since 
		\begin{equation*}
			\E_Y \left[ r(Y)^2 \right] = \E_X(1-p(X))^2 = (1- \mathbf{E}_1)^2
		\end{equation*}
		and therefore
		\begin{equation*}
			\E_Y (r(Y) - (1- \mathbf{E}_1))^2 = 0.
		\end{equation*}
		Analogously is shown for $q(X)$. Therefore $p(X)$ almost surely not constant implies, $q(X)$ and $r(Y)$ are almost surely not constant. With the condition, $r(Y)$ a.s. not constant we have the lower bound (\ref{eqn::sigma_estimate}).\\}
	
	\hfill$\Box$
	
	\subsection{Proof of Theorem \ref{thm::consistency}}
	\noindent
	Under a fixed alternative for which it holds $\mathbf{E}_1 \neq \frac{1}{2}$, due to Theorem \ref{thm::Tconv} we have
	\begin{equation*}
		\left| \frac{1}{ \# I^{(n,m)}}T^{(n,m)} -\frac{1}{2} \right| \overset{a.s.}{\underset{n,m \rightarrow \infty}{\rightarrow}} \left| \mathbf{E}_1 - \frac{1}{2} \right| > 0
	\end{equation*}
	The order of $\Var_{H}(T^{(n,m)})$, variance of $T^{(n,m)}$ under null hypothesis, is the same as the order of $\# C_{i_1 = j_1}$ (cf. proof of Theorem \ref{thm::T}) and there exists some constant $C>0$ s.t.
	\begin{equation*}
		\Var_{H}(T^{(n,m)} )< C \# C_{i_1 = j_1}.
	\end{equation*}
	On the other hand, with $\# C_{i_1 = j_1} \sim n^5m^2$ and $\# I^{(n,m)} \sim n^3m$
	\begin{align*}
		\left|\frac{T^{(n,m)} - \frac{1}{2} \# I^{(n,m)} }{\sqrt{\Var_{H}(T^{(n,m)} )}} \right| \geq  n^{\frac{1}{2}} \left| \frac{1}{ \# I^{(n,m)}}T^{(n,m)} -\frac{1}{2} \right| C \overset{a.s.}{ \underset{n,m \rightarrow \infty}{ \rightarrow}} \infty 
	\end{align*}
	with a positive constant $C$ and for the $\alpha$-Level critical value $c_{\frac{\alpha}{2}}$  hence
	\begin{equation*}
		\Prob \left(\left|\frac{T^{(n,m)} - \frac{1}{2} \# I^{(n,m)} }{\sqrt{\Var_{H}(T^{(n,m)} )}} \right| > c_{\frac{\alpha}{2}} \right)\underset{n,m \rightarrow \infty}{ \rightarrow} 1.
	\end{equation*}

	\hfill$\Box$
	
	\subsection{Combinatorial Quantities} \label{subsec::combinatorial_quantities}
	
	In this subsection we count the quantities corresponding to the index sets and the local dependency structure of the triangular arrays in Section \ref{Sec:Motivation}, i.e. find the order of $\# N_{n,m}(\mathbf{i})$, the degree of a vertex $\mathbf{i}$, necessary for the estimations figuring in Theorem \ref{thm::CLT}. These quantities only depend on the sizes of two groups $n$ and $m$, where we assume that $m$ is large enough to perform the independent sampling described in Section \ref{Sec:Motivation}. For the sake of completeness, we also calculate the combinatorial constants in (\ref{eqn::TT}) and the sizes of the related sets.
	
	\noindent
	We firstly count the number of elements in the index sets $I_1^{(n,m)}$ and $I_2^{(n,m)}$. $\binom{n}{2}$ is the number of ways we can pick two points from the point cloud $\mathcal{X}$ and we combine this choice with $(n-2)m$ ways we can combine one of the remaining $\mathcal{X}$-points with one of the $\mathcal{Y}$-points, having 
	\begin{equation}\label{eqn::NoI}
		\#I_1^{(n,m)} =\binom{n}{2} (n-2)m.  
	\end{equation}
	Similarly, in order to choose two between-clouds distances, we can pick two points from  $\mathcal{X}$ and two points from $\mathcal{Y}$ in $\binom{n}{2}$ and $\binom{m}{2}$ and we can connect them in two ways, concluding
	\begin{equation*}
		\#I_2^{(n,m)} = \frac{1}{2}nm(n-1)(m-1). 
	\end{equation*} 
	For  $I^{(n,m)}  = I_1^{(n,m)} \cup I_2^{(n,m)}$ it is therefore
	\begin{equation*}
		\# I^{(n,m)}  = \frac{1}{2} n^3 m + n^2 m^2 - n m^2 - \frac{5}{2}n^2 m + 2nm.
	\end{equation*}
	\noindent
	We further estimate the size of the dependency neighborhoods (degree of the vertices) of $\mathbf{i} \in I_1^{(n,m)}$ and $\mathbf{j} \in I_2^{(n,m)}$ in the dependency graph $G_{n,m}$ generated by the triangular array defined in (\ref{eqn::triangle}). We notice that all the vertices of a type ($I_1^{(n,m)}$ or $I_2^{(n,m)}$) have the same degree. In both cases, two vertices $\mathbf{i}$ and $\mathbf{j}$ are connected iff $\mathbf{i} \cap \mathbf{j} \neq \emptyset$ and multiincides $\mathbf{i}$ and $\mathbf{j}$ can share up to $3$ elements. The highest number of vertices connected to $\mathbf{i}$ shares only one element of $\mathbf{i}$ and the number of connections with a larger intersection is of lesser order than this number.\\ 
	\noindent
	We  count $\# N_{n,m}(\mathbf{i})$ for $\mathbf{i} \in I_1^{(n,m)}$. Firstly we count the number of elements in the complementary set, i.e. for the fixed $\mathbf{i}$ we count the number of from $\mathbf{i}$ independent vertices in $I_1^{(n,m)}$ and in $I_2^{(n,m)}$. In the first case, there are $\binom{n-3}{2} (n-5) (m-1)$ possibilities to pick $\mathbf{j} \in I_1^{(n,m)}$ by excluding indices in $\mathbf{i}$. For $I^{(n,m)}= I_1^{(n,m)}$ it is then
	\begin{equation}\label{eqn::NoN1}
		\# N_{n,m}(\mathbf{i}) = \frac{n^3}{2} + \frac{9m n^2}{2} - 6 n^2  -\frac{45mn}{2} +\frac{47n}{2} + 30 m  -30
	\end{equation}
	We further consider $I^{(n,m)} : = I_1^{(n,m)} \cup I_2^{(n,m)}$. There are $(n-3)(m-1)(n-4)(m-2)$ possibilities to pick $\mathbf{j} \in I_2^{(n,m)}$ with $\mathbf{i} \cap \mathbf{j} = \emptyset$. Then since it is $\# N_{n,m}(\mathbf{i})  = \# I - \binom{n-3}{2} (n-5) (m-1) - (n-3)(m-1)(n-4)(m-2)$ we calculate
	\begin{equation*}
		\# N_{n,m}(\mathbf{i})  = \frac{1}{2}n^3 + \frac{13}{2}n^2 m + 6n m^2 -8 n^2 -12 m^2 - \frac{85}{2}nm + \frac{75}{2} n +66 m -54.
	\end{equation*}
	Similarly, regarding $N_{n,m}(\mathbf{i})$ for $\mathbf{i} \in I_2^{(n,m)}$, we consider the complementary set, the set of all $\mathbf{j}$ with $\mathbf{i} \cap \mathbf{j} = \emptyset$. In the case $\mathbf{j} \in I_1^{(n,m)}$ we have $\binom{n-2}{2} (n-4)(m-2)$ and in the case $\mathbf{j} \in I_2^{(n,m)}$ we have $(n-2)(m-2)(n-3)(m-3)$ possibilities. We conclude that the cardinality of the dependency neighbourhood of $\mathbf{i}\in I_2^{(n,m)}$ is
	\begin{equation}\label{eqn::NoN2}
		\# N_{n,m}(\mathbf{i})  = n^3 + 7 n^2 m + 4nm^2 -15n^2 -36nm - 6m^2 +56n + 42 m -60.
	\end{equation}
	\\
	\noindent
	We can determine the cardinality of $\mathcal{D}^{(n,m)}$ by combining all the choices of $2r$ pairwise connected points from the first group with all the choices of connections between the rest $n-2r$ points of the first group with $m$ points of the second, where $r:= \lfloor \frac{n}{3} \rfloor$. We therefore have
	\begin{equation}
		\#\mathcal{D}^{(n,m)}= \binom{n}{2r} \multiset{2r}{2} \binom{m}{n-2r} (n-2r)! 
	\end{equation}
	where $ \multiset{2r}{2}  := \Pi_{i=0}^{r-1} \binom{2(r-i)}{2}$. This can be further simplified as $\multiset{2r}{2}  = \frac{(2r)!}{2^r}$. 
	
	We proceed to calculate the number of elements in $\mathcal{D}^{(n,m)}$ as
	\begin{equation*}
		\#\mathcal{D}^{(n,m)}=\frac{n! m!}{2^r (n-2r)! (m-n+2r)!}
	\end{equation*}
	By fixing a pair of tupels $\mathbf{i}$ (i.e. by fixing some four vertices $X_{i_1}$, $X_{i_2}$, $X_{i_3}$ and $Y_{i_4}$), we analogously have 
	\begin{equation*}
		\# \mathfrak{d}_1 = \binom{n-3}{2(r-1)} \multiset{2(r-1)}{2} \binom{m-1}{n-2r-1} (n-2r-1)!
	\end{equation*}
	The constant $C_1^{(n,m)}$ simplifies to 
	\begin{equation}\label{eqn::C1}
		C_1^{(n,m)} = \frac{2(n-2r)}{mn(n-1)(n-2)}
	\end{equation}
	Similarly, by fixing two pairs $X_{i_1}$, $Y_{i_2}$, $X_{i_3}$, $Y_{i_4}$, we have
	\begin{equation*}
		\# \mathfrak{d}_2 = \binom{n-2}{2r} \multiset{2r}{2} \binom{m-2}{n-2r-2} (n-2r-2)!
	\end{equation*}
	which reduces the constant $C_2^{(n,m)}$ to 
	\begin{equation}\label{eqn::C2}
		C_2^{(n,m)}= \frac{(n-2r)(n-2r-1)}{nm(n-1) (m-1)}.
	\end{equation}
	
	\subsection{Numerical estimation of $P_1$ over distributions, metrics and dimensions}\label{sec:numerical P1}
	
	\noindent
	The probability (\ref{eqn::P1}) has been estimated by $10^6$ Monte Carlo samplings for concrete distributions, metrics and dimensions as given in Tables \ref{tab::Normal}-\ref{tab::Cauchy}. The exemplified distributions include standard normal, uniform, Gamma with shape parameter $k= 1$ and scale parameter $\theta = 2$ and Cauchy distribution with location $0$ and scale $1$. Particularly the heavy-tailed distributions were of interest, since the concentration of the distribution determines $P_1$. The dimensions are always independent. Note that the estimated $P_1$ values hold true for any scaling or translation of the examined distributions, given the independence of the dimensions, due to (\textbf{D1}), except for the Canberra distance. Canberra refers to Canberra distance function given with $d(X,Y) = \sum_{i=1}^D \frac{|X_i - Y_i|}{|X_i|+|Y_i|}$, which is of special interest.\\
	
	\begin{table}[hbt!]
		\centering
		\begin{tabular}{ |c|c|c|c|c|c| } 
			\hline
			Metric\textbackslash Dimension & 1& 2 & 5 & 10 & 50 \\ 
			\hline
			$l^1$ & 0.528276 & 0.532446 &  0.534103 & 0.534561 & 0.534962  \\ 
			\hline
			$l^2$ & 0.528205 & 0.533552 &  0.536747 & 0.538320 & 0.533520\\ 
			\hline
			$l^5$ & 0.527722 & 0.533157  & 0.535111 & 0.533912 & 0.530244\\ 
			\hline
			$l^{\infty}$ & 0.527916 & 0.531780  & 0.531696 & 0.528262 & 0.517773\\ 
			\hline
			Canberra & 0.532763 & 0.506087 & 0.505927 & 0.506346 & 0.506765 \\ 
			\hline
		\end{tabular}
		\caption{Estimated probability (\ref{eqn::P1}) for $\mathcal{N}(\mathbf{0}_D, \mathbf{I}_D)$}
		\label{tab::Normal}
	\end{table}
	
	\begin{table}[hbt!]
		\centering
		\begin{tabular}{ |c|c|c|c|c|c| } 
			\hline
			Metric\textbackslash Dimension & 1 & 2 & 5 & 10 & 50  \\ 
			\hline
			$l^1$ & 0.510712 & 0.514656 & 0.514466  & 0.515283 & 0.516251  \\ 
			\hline
			$l^2$ & 0.511553 & 0.517824 & 0.521306  & 0.522040 & 0.522902\\ 
			\hline
			$l^5$ & 0.510950& 0.519395 & 0.523481 & 0.524618 & 0.524708 \\ 
			\hline
			$l^{\infty}$ & 0.510413 & 0.518092 & 0.521787 & 0.518908 & 0.510352 \\ 
			\hline
			Canberra & 0.536694 & 0.540395 & 0.540820 & 0.541016 & 0.541079 \\ 
			\hline
		\end{tabular}
		\caption{Estimated probability (\ref{eqn::P1}) for $\mathcal{U}(0, 1)^D$}
		\label{tab:: Uniform}
	\end{table}
	
	\begin{table}[hbt!]
		\centering
		\begin{tabular}{ |c|c|c|c|c|c| } 
			\hline
			Metric\textbackslash Dimension  & 1 & 2 & 5 & 10 & 50  \\ 
			\hline
			$l^1$ & 0.536443 & 0.544270 &  0.549540 & 0.551535 & 0.552521   \\ 
			\hline
			$l^2$ & 0.535988 & 0.545402 &  0.553075  & 0.557083 & 0.562266 \\ 
			\hline
			$l^5$ & 0.536178 & 0.544418 & 0.551290 &  0.555977 &  0.561372\\ 
			\hline
			$l^{\infty}$ & 0.536259 & 0.544070 & 0.550366 & 0.552803 &  0.555821 \\ 
			\hline
			Canberra &0.532183 & 0.531522 & 0.530954 & 0.530633 & 0.530991 \\ 
			\hline
		\end{tabular}
		\caption{Estimated $P_1$ for $\Gamma(1, 2)^D$}
		\label{tab::Gamma}
	\end{table}
	
	\begin{table}[hbt!]
		\centering
		\begin{tabular}{ |c|c|c|c|c|c| } 
			\hline
			Metric\textbackslash Dimension  & 1 & 2 & 5 & 10 & 50  \\ 
			\hline
			$l^1$ & 0.556312 & 0.564395 &  0.570460 & 0.573040 & 0.576032   \\ 
			\hline
			$l^2$ & 0.557003 & 0.564242 &  0.569137  & 0.571805 & 0.572302 \\ 
			\hline
			$l^5$ & 0.555691 & 0.563025 & 0.568439 &  0.570655 &  0.570892\\ 
			\hline
			$l^{\infty}$ & 0.556125 & 0.564030 & 0.567522 & 0.569805 &  0.571181 \\ 
			\hline
			Canberra & 0.532626 & 0.507481 & 0.508951 & 0.508473 & 0.506963 \\ 
			\hline
		\end{tabular}
		\caption{Estimated $P_1$ for Cauchy distribution with location $0$ and scale $1$}
		\label{tab::Cauchy}
	\end{table}
	
	\begin{table}[hbt!]
		\centering
		\begin{tabular}{ |c|c|c|c|c|c| } 
			\hline
			Metric\textbackslash Dimension  & 1 & 2 & 5 & 10 & 50  \\ 
			\hline
			$l^1$ & 0.556344  & 0.563713 & 0.568862  & 0.571261 & 0.574518   \\ 
			\hline
			$l^2$ & 0.556311 & 0.562560 & 0.568479   & 0.570540 & 0.573249 \\ 
			\hline
			$l^5$ & 0.554668 &  0.562374 & 0.567011 &  0.568981 & 0.570871 \\ 
			\hline
			$l^{\infty}$ & 0.557274 & 0.561351 & 0.568047 & 0.568846 & 0.570900  \\ 
			\hline
			Canberra & 0.536604 &  0.540105 & 0.540209 &  0.540613 &   0.542141\\ 
			\hline
		\end{tabular}
		\caption{$P_1$ for Pareto with location $1$ and shape $1$}
		\label{tab::Pareto}
	\end{table}

	\end{document}